%% file: Power_excitation_v01c.tex
\documentclass[journal]{IEEEtran}  
\newcommand\figwidth{0.9}

\usepackage{amsmath} 
\usepackage{amssymb}  
\usepackage{amsthm}
\usepackage{graphicx}
\usepackage{enumitem}
\usepackage{hyperref}
\usepackage{cite}
\usepackage{color}

\newtheorem{theorem}{Theorem}
\newtheorem{lemma}{Lemma}
\newtheorem{definition}{Definition}
\newtheorem{corollary}{Corollary}
\newtheorem{remark}{Remark}

\usepackage{amsmath} 
\usepackage{amssymb}  
\usepackage{amsthm}
\usepackage{color}
\usepackage{graphicx}
\usepackage{enumitem}
\usepackage{hyperref}
\usepackage{cite}
\usepackage{dsfont}
\makeatletter
\newenvironment{subtheorem}[1]{%
  \def\subtheoremcounter{#1}%
  \refstepcounter{#1}%
  \protected@edef\theparentnumber{\csname the#1\endcsname}%
  \setcounter{parentnumber}{\value{#1}}%
  \setcounter{#1}{0}%
  \expandafter\def\csname the#1\endcsname{\theparentnumber\alph{#1}}%
  \expandafter\def\csname theH#1\endcsname{thm.\theparentnumber\alph{#1}}%
  \unskip\ignorespaces
}{%
  \setcounter{\subtheoremcounter}{\value{parentnumber}}%
  \ignorespacesafterend
}
\makeatother
\newcounter{parentnumber}

\usepackage[ruled,vlined]{algorithm2e}
\SetKwInput{KwInput}{Input}                
\SetKwInput{KwOutput}{Output}              

\newtheorem{conjecture}{Conjecture}

\title{\LARGE \bf
Multipliers for forced Lurye systems  with slope-restricted nonlinearities}

\author{William P. Heath,~\IEEEmembership{Member,~IEEE,} Sayar Das and Joaquin Carrasco
\thanks{William P. Heath is with the School of Computer Science and Engineering, Bangor University, UK, 
       {\tt\small w.heath@bangor.ac.uk}
       
       Sayer Das and Joaquin Carrasco are with the Control Systems and Robotics Group, Department of Electrical and Electronic Engineering, University of Manchester, UK, {\tt\small sayar.das@postgrad.manchester.ac.uk,}
        {\tt\small joaquin.carrascogomez@manchester.ac.uk}}%
        }

\begin{document}

\maketitle
\thispagestyle{empty}
\pagestyle{empty}

\begin{abstract}

Dynamic multipliers can be used to guarantee the stability of Lurye systems with slope-restricted nonlinearities, but give no guarantee that the closed-loop system has finite incremental gain. We show that multipliers guarantee the closed-loop power gain to be bounded and quantifiable. Power may be measured about an appropriate steady state bias term, provided the multiplier does not require the nonlinearity to be odd. Hence dynamic multipliers can be used to guarantee such  Lurye systems have low sensitivity to noise, provided other exogenous signals have constant steady state. For periodic excitation,  the closed-loop response can apparently have a subharmonic or chaotic response. We revisit a class of multipliers that can guarantee a unique, attractive and period-preserving solution.  We show the multipliers can be derived using classical tools and reconsider assumptions required for their application.  Their phase limitations are inherited from those of discrete-time multipliers. The multipliers cannot be used at all frequencies unless the circle criterion can also be applied; this is consistent with known results about dynamic multipliers and incremental stability.

\end{abstract}



\section{Introduction}
The O'Shea-Zames-Falb (OZF) multipliers preserve the positivity of bounded monotone nonlinearities and hence can be used to guarantee the absolute stability (with bounded input-output $\mathcal{L}_2$ gain) of Lurye systems with slope-restricted nonlinearities \cite{Turner12,Veenman14,Bertolin25}. They were proposed by O'Shea \cite{OShea67} and formalised by Zames and Falb \cite{Zames68}. Modern tools can search for suitable multipliers and give an upper bound on the $\mathcal{L}_2$ gain \cite{Megretski04,Kao04,Turner12,Veenman16}. They remain the widest known class of multiplier for slope-restricted nonlinearities \cite{Carrasco:13}. 
Their discrete-time counterparts were formalised by Willems and Brockett \cite{Willems68}. 
Efficient searches for discrete-time multipliers are proposed in \cite{Wang:TAC}.

It has been argued in the literature, both specifically with respect to Lurye systems \cite{Kulkarni:2002a,Waitman17} and more generally \cite{Zames66a,Fromion96,Angeli02,Sepulchre22,Chaffey23}, that emphasis should be given to finite incremental gain. In particular for forced systems, if exogenous signals are not in $\mathcal{L}_2$ then desirable properties one might infer for linear systems do not necessarily carry over to nonlinear systems. Zames \cite{Zames66a} argues
 that a definition of closed-loop stability should require both continuity and boundedness, inter alia so that outputs are not ``critically sensitive to small changes in inputs --- changes such as those caused by noise''; Theorem 3 in \cite{Zames66a} gives conditions on the incremental positivity of the loop elements that are sufficient to achieve this.
Similarly it was known at the time \cite{Brockett66a} that, for Lurye systems with slope-restricted nonlinearites, the circle criterion could be used to guarantee ``the existence of unique steady-state oscillations (= absence of ``jump phenomena'' and subharmonics) in forced nonlinear feedback systems''. It was subsequently established \cite{Kulkarni:2002a} that dynamic multipliers do not, in general, preserve the incremental positivity of nonlinearities.
As noted in \cite{Kulkarni:2002a} there is some irony that the definition of stability used by Zames and Falb in \cite{Zames68} does {\em not} require  closed-loop continuity.

There remain open questions, for Lurye systems with slope-restricted nonlinearities,  about what behaviour is and is not guaranteed by the OZF multipliers when the exogenous signal is a power signal outside $\mathcal{L}_2$. In this paper we consider three classes of exogenous signals:
\begin{itemize}
    \item small power signals (for example noise);
    \item signals with constant steady state (for example Heaviside step signals);
    \item periodic signals (for example sine waves).
\end{itemize}

Certainly lack of finite incremental gain can lead to undesirable effects.
An example  of  such  a Lurye system where finite-gain stability is guaranteed but small changes in input can lead to significant changes in output is given in \cite{Fromion04}. It is shown in \cite{Fromion04} that this example has at least two attractive limit cycles when the excitation is periodic. We show further that this example with a different nonlinearity may lead to a subharmonic response. We discuss an additional discrete-time example in this paper which may have a chaotic response to periodic excitation.

However, we argue that finite incremental gain is not necessary to ensure insensitivity to noise signals. In particular the existence of a suitable OZF multiplier can be used to guarantee that such a Lurye system is insensitive to noise for a wide class of exogenous signal. This is timely in that OZF multipliers have recently been proposed for the design of control systems with saturation \cite{Veenman14,Bertolin22}. In addition, we show that, provided we do not exploit any oddness of the nonlinearity, the existence of a suitable OZF multiplier guarantees a unique steady-state input-output map.

We observe more generally that finite-gain stability ensures small power noise input leads to small power output when all other exogenous signals are in $\mathcal{L}_2$. 
 We define a notion of finite-gain offset stability 
 and observe similarly that finite-gain offset stability ensures small power noise input leads to small power output when all other exogenous signals are bias signals. Such properties can be guaranteed for Lurye systems with slope-restricted nonlinearities when there is a suitable OZF multiplier. 
 




With respect to periodic signals, Altshuller \cite{Altshuller:11,Altshuller:13} defines a subclass of OZF multipliers 
that can be used to guarantee  such  a Lurye system, subject to a non-zero exogenous signal with period $T$, has a unique solution also with period $T$ that is a global attractor. We will denote members of this subclass as the Altshuller multipliers.

We revisit the approach of Altshuller \cite{Altshuller:11,Altshuller:13}. His analysis requires the assumption that a periodic solution (not necessarily an attractor) exists \cite{Rasvan11}. We show that such an assumption is justified when the Lurye system has a state-space representation (either finite dimensional or delay-differential).
Altshuller shows that his multipliers preserve the positivity of periodic nonlinearities via the restrictive delay-IQC (integral quadratic constraint) approach. We show that result can be obtained straightforwardly via classical analysis or the IQC approach of \cite{Megretski97}; furthermore we show that no OZF multiplier outside Altshuller's subclass shares this property. We explore the relation between the Altshuller multipliers and both the continuous-time OZF multipliers and their discrete-time counterparts.  In particular we show that the Altshuller multipliers inherit the phase properties of the discrete-time OZF multipliers.  Finally, for a given Lurye system  with a slope-restricted nonlinearity, we show that there only exist suitable Altshuller multiplers for all periods  when the circle criterion can also be used to establish incremental stability.

While our development is for continuous-time single-input single-output systems, results can be straighforwardly generalised to both discrete-time and multivariable systems. One of the examples we discuss is discrete-time. Preliminary results for noise and bias signals were presented in \cite{Heath:24}.



The incremental stability properties of Lurye systems have been extensively studied in the literature. In \cite{MIRANDAVILLATORO2018}, contraction and p-dominance properties of Lurye systems are employed to analyse the existence and stability of attracting orbits.
Incremental versions of the integral quadratic constraint (IQC) framework are proposed in \cite{Jonsson03} and more recently \cite{Su2025}. In \cite{Jonsson03} incremental IQCs are used to establish existence, uniqueness and attractiveness of periodic orbits; then standard IQCs can be used to analyse robust performance. In \cite{Su2025} incremental properties of Lurye systems are analysed using a class of dynamic multipliers. However, as noted previously, the OZF class of multipliers does not preserve the incremental positivity of slope-restricted nonlinearities \cite{Kulkarni2002}. One possible remedy is to restrict the class of systems under consideration; for example, \cite{Drummond2025} investigates the incremental stability of Lurye interconnections between externally positive systems and the class of incremental gain systems. 

\section{Preliminaries}

 \subsection{Signals and systems}

 Let $\mathcal{L}_2$ be the space of finite energy Lebesgue integrable signals on $[0,\infty)$  with norm
 \begin{equation}
 \|y\| = \left ( \int_0^{\infty}y(t)^2 \,dt\right )^{\frac{1}{2}}.
 \end{equation}
 Let  $\mathcal{L}_{2e}$ be the corresponding extended space (see for example  \cite{desoer75}). The truncation $y_T\in\mathcal{L}_2$ of $y\in\mathcal{L}_{2e}$ is given by
 \begin{equation}
     y_T(t) = \left \{ \begin{array}{lll}y(t) & \text{for} & 0\leq t \leq T,\\
     0 & \text{for} & T<t.\end{array}\right .
 \end{equation}
 \begin{definition}
 Let $\mathcal{P}\subset\mathcal{L}_{2e}$ be the space of finite power locally Lebesgue integrable signals on $[0,\infty)$ with seminorm
 \begin{equation}\label{def_P}
 \|y\|_P = \left (\limsup_{T\rightarrow\infty}\frac{1}{T}\|y_T\|^2\right )^{\frac{1}{2}}.
 \end{equation}
 We say $y$ is a {\bf power signal} if $y\in\mathcal{P}$.
 Let $\mathds{1}\in\mathcal{P}$ be the Heaviside step function given by 
     $\mathds{1}(t)=1$ for all  $t>0$.
 Define the {\bf bias} $\bar{y}\in\mathbb{R}$ of  a signal $y\in\mathcal{P}$ as
\begin{equation}
    \bar{y} = 
    \arg \min_{\bar{y}\in\mathbb{R}}\|(y-\bar{y}\mathds{1})\|_P.
\end{equation}
We say $y$ is an $\mathcal{L}_2$-{\bf bias signal with bias} $\bar{y}$ if $\bar{y}$ is unique and $y-\bar{y}\mathds{1}\in\mathcal{L}_2$. 

A signal is {\bf periodic with period} $T$ if $y(t+T)=y(t)$ for all $t\geq 0$. Let $\mathcal{S}_T\subset\mathcal{P}$ be the class of signal that can be expressed as $y=y_1+y_2$ with $y_1$ periodic with period $T$ and $y_2\in\mathcal{L}_2$.

\end{definition}

\begin{remark}
The limit superior in (\ref{def_P}) does not appear in standard definitions of power (e.g. \cite{vidyasagar93,Zhou96}) but is necessary to ensure $\mathcal{P}$ is a vector space \cite{Partington04,Mari96}\footnote{We are grateful to Andrey Kharitenko for this observation.}.
\end{remark}

 A map $\boldsymbol{H}:\mathcal{L}_{2e}\rightarrow\mathcal{L}_{2e}$ is stable if $u\in\mathcal{L}_2$  implies $\boldsymbol{H}(u)\in\mathcal{L}_2$. It is finite-gain stable (FGS) if there is some $h<\infty$ such that $\|\boldsymbol{H}(u)\|\leq h \|u\|$ for all $u \in \mathcal{L}_2$. Its gain is the smallest such~$h$.

\begin{remark}\label{rem:ic1}
 Our definition of finite-gain stability carries the assumption that we have zero initial conditions.
 Non-zero initial conditions can be accommodated provided they can be represented with a nonlinear state-space description that is reachable and uniformly observable \cite{vidyasagar93}. 
 \end{remark}

\begin{definition}\label{def:offset}
Let $\boldsymbol{H}:\mathcal{L}_{2e}\rightarrow\mathcal{L}_{2e}$. We say $\boldsymbol{H}$ is {\bf offset stable} if there is some function $H_0:\mathbb{R}\rightarrow\mathbb{R}$ such that if $u$ is an $\mathcal{L}_2$-bias signal with bias $\bar{u}$ then $\boldsymbol{H}(u)$ is an $\mathcal{L}_2$-bias signal with bias $H_0(\bar{u})$.
We call $H_0$ the steady state map of $\boldsymbol{H}$.
Define $\boldsymbol{H_{\bar{u}}}:\mathcal{L}_{2e}\rightarrow\mathcal{L}_{2e}$ as
\begin{equation}
    \boldsymbol{H_{\bar{u}}}(u)=\boldsymbol{H}(u+\bar{u}\mathds{1})-H_0(\bar{u}).
\end{equation}
It follows that $\boldsymbol{H}$ is offset stable if $\boldsymbol{H_{\bar{u}}}$ is stable for all $\bar{u}\in\mathbb{R}$.
We say $\boldsymbol{H}$ is {\bf finite-gain offset stable} (FGOS) if there is some $h<\infty$ such that $\boldsymbol{H_{\bar{u}}}$ is FGS  with gain less than or equal to $h$ for all $\bar{u}\in\mathbb{R}$. We call the minimum such $h$ the {\bf offset gain} of $\boldsymbol{H}$.
\end{definition}

 \subsection{Continuity}
 \begin{definition}[\cite{Zames66a}]
 The {\bf incremental gain} of $H:\mathcal{L}_{2e}\rightarrow\mathcal{L}_{2e}$ is the supremum of 
 $\|(H(x))_T-(H(y))_T\| / \|x_T-y_T\|$
 over all $x,y\in\mathcal{L}_{2e}$ and all $T>0$ for which $\|x_T-y_T\|\neq 0$.
 \end{definition}
 Finite-gain stability does not guarantee finite incremental gain: if $H$ is FGS and $u_1,u_2\in\mathcal{L}_2$ the ratio $R=\|H(u_1)-H(u_2)\|/\|u_1-u_2\|$ may be arbitrarily large. In particular, suppose $v_1,v_2$ are power signals with $v_1-v_2\in\mathcal{L}_2$ but $v_1,v_2\notin\mathcal{L}_2$ and suppose $H$ is FGS but $H(v_1)-H(v_2)\notin\mathcal{L}_2$. Let $u_1,u_2$ be the truncations $u_1=(v_1)_T$ and $u_2=(v_2)_T$. Then $R\rightarrow\infty$ as $T\rightarrow\infty$. In \cite{Heath:24} we discussed a discrete-time example that is FGS where specific inputs $v_1,v_2$ satisfy $v_1-v_2\in \ell_2$ but $H(v_1)-H(v_2)\notin \ell_2$.

 As noted in the Introduction, dynamic multipliers do not, in general, preserve the incremental positivity of nonlinearities \cite{Kulkarni:2002a}. An example of a Lurye system where multipliers guarantee finite gain stability but where continuity of the input-output map is lost  is given in \cite{Fromion04} and discussed further below.

\subsection{State space stability}

We follow the terminology of \cite{Yoshizawa66} for systems whose dynamics are driven by ordinary differential equations  
with state $x(t;t_0,x_0)\in\mathbb{R}^n$, with initial condition $x(t_0)=x_0$ and where
\begin{equation}
    \frac{dx}{dt}=F(t,x)
\end{equation}
for some $F:\mathbb{R}^+\times\mathbb{R}^n\rightarrow\mathbb{R}^n$ with $F(t,0)=0$ for all $t\geq t_0$ and $F$ continuous on the solution space. Let $|x(t)|$ be the 2-norm of $x(t)$.
\begin{definition}
        The {\bf zero solution} is the solution $x(t)=0$ for all $t\geq t_0$ when $x(t_0)=0$.
        Solutions are {\bf uniform-bounded} if for any $\alpha>0$ and $t_0\in[0,\infty)$ there exists a $\beta(\alpha)>0$ such that if $|x_0|\leq \alpha$, $|x(t;x_0,t_0)|<\beta(\alpha)$ for all $t\geq t_0$.
        Solutions are {\bf ultimately bounded for bound} $B$ if for every solution $x(\cdot;x_0,t_0)$ there exists a $t_s>0$ such that $|x(t;x_0,t_0)|<B$ for all $t\geq t_0+t_s$.
        Solutions are {\bf equi-ultimately bounded for bound }$B$ if for any $\alpha>0$ and $t_0\geq 0$ there exists a $t_s(t_0,\alpha)$ such that if $|x_0|<\alpha$ then $|x(t;t_0,x_0)|<B$ for all $t\geq t_0+t_s(t_0,\alpha)$.
        Solutions are {\bf uniform-ultimately bounded for bound }$B$ if they are equi-ultimately bounded for bound $B$ with $t_s$ independent of $t_0$.
        The zero solution is {\bf stable} if for any $\varepsilon>0$ and any $t_0\in [0,\infty)$ there exists a $\delta(t_0,\varepsilon)>0$ such that if $|x_0|<\delta(t_0,\varepsilon)$ we have $|x(t;t_0,x_0)|<\varepsilon$ for all $t\geq t_0$.
        The zero solution is {\bf uniform-stable} if it is stable with $\delta$ independent of $t_0$.
        The zero solution is {\bf asymptotically stable in the large} if it is stable and if every solution tends to zero as $t\rightarrow\infty$.
        The zero solution is {\bf quasi-uniform-asymptotically stable in the large} if for any $\alpha>0$, any $\varepsilon>0$ and any $t_0\in [0,\infty)$ there exists a $t_s(\varepsilon,\alpha)>0$ such that for any $t_0\in [0,\infty)$ if $|x_0|<\alpha$ then $|x(t;x_0,t_0)|<\varepsilon$ for all $t\geq t_0+t_s(\varepsilon,\alpha)$.
        The zero solution is {\bf uniform-asymptotically stable in the large} if it is uniform-stable, quasi-uniform asymptotically stable in the large and solutions are uniform-bounded.
\end{definition}
Similar properties can be defined for systems whose dynamics are driven by delay-differential equations \cite{Yoshizawa66}.

\subsection{Memoryless, monotone and bounded operators}
\begin{definition}
Let $\Phi^{m}$ be the class of {\bf memoryless, monotone and bounded (MMB) operators} $\boldsymbol{\phi}:\mathcal{L}_{2e}\rightarrow\mathcal{L}_{2e}$ that, for every input signal $u\in\mathcal{L}_{2e}$, can be characterised by some monotone and bounded function $N:\mathbb{R}^+\times\mathbb{R}\rightarrow\mathbb{R}$ with 
\begin{equation}(\boldsymbol{\phi}(u))(t) = N(t,u(t)), \text{ for all } t\ge0.
\end{equation}
$N$ is monotone (in the second variable) in the sense that 
$N(t,x_1)\geq N(t,x_2)$ for all $t\geq 0$ and $x_1\geq  x_2$. $N$ is bounded (in the second variable) in the sense that there exists a $C\geq 0$ such that $|N(t,x)|\leq C|x|$ for all $t\geq 0$ and $x\in\mathbb{R}$. We say $N\in \mathcal{N}^{m}$ if $N:\mathbb{R}^+\times\mathbb{R}\rightarrow\mathbb{R}$ characterises some $\mathbf{\phi}\in\Phi^{m}$. We say $\mathcal{N}^{m}$ characterises $\Phi^{m}$.



Let $\Phi^{sr}_k\subset\Phi^{m}$ be the class of {\bf slope-restricted} (on $[0,k]$)  MMB operators, characterised by $\mathcal{N}^{sr}_k\subset\mathcal{N}^{m}$ whose members $N\in\mathcal{N}^{sr}_k$ are slope-restricted on $[0,k]$ in the sense that they satisfy $0\leq (N(t,x_1) - N(t,x_2))/(x_1-x_2)\leq k$ for all $t\geq 0$ and  $x_1\neq x_2$.

Let $\Phi^{ti}\subset\Phi^{m}$ be the class of {\bf time-invariant} MMB operators, characterised by $\mathcal{N}^{ti}\subset\mathcal{N}^m$ whose members $N\in\mathcal{N}^{ti}$ satisfy $N(t_1,x)=N(t_2,x)$ for all $t_1$, $t_2\geq 0$ and $x\in\mathbb{R}$. If $N(t,x)$ characterises a time-invariant MMB operator we can define $Q:\mathbb{R}\rightarrow\mathbb{R}$ such that $Q(x)=N(t,x)$ for all $x\in\mathbb{R}$ and $t\geq 0$. With some abuse of notation we will say a time-invariant MMB operator is characterised by such a $Q$, belonging to the class~$\mathcal{Q}$. 

Let $\Phi^p_T\subset\Phi^m$ be the class of {\bf  periodic} MMB operators (with period $T$), characterised by $\mathcal{N}^p_T\subset\mathcal{N}^m$ whose members $N\in\mathcal{N}^p_T$ satisfy $N(t+T,x)=N(t,x)$ for all $t\geq 0$ and $x\in\mathbb{R}$. 

Let $\Phi^{odd}\subset\Phi^m$ be the class of {\bf odd} MMB operators 
 characterised by $\mathcal{N}^{odd}\subset\mathcal{N}^m$ whose  members $N\in\mathcal{N}^{odd}$ satisfy $N(t,x)=-N(t,-x)$ for all $t\geq 0$ and $x\in\mathbb{R}$.
\end{definition}

 \subsection{Linear operators}

 
  Following \cite{vidyasagar93} we define a class of LTI (linear time invariant) and stable operators as follows.
\begin{definition}
Let $\mathcal{S}$ be the class of  continuous time convolution operator{\color{red}s} $\boldsymbol{G}:\mathcal{L}_{2e}\rightarrow\mathcal{L}_{2e}$ whose impulse response takes the form
\begin{equation}
g(t) = \left \{
            \displaystyle{\begin{array}{ll}
                0 & \text{when }t<0,\\
                \displaystyle{\sum_{i=0}^{\infty}}g_i\delta(t-t_i)+g_a(t) & \text{when }t\geq0,
            \end{array}}
        \right .
\end{equation}
where $\delta(\cdot)$ denotes the unit delta distribution, $0\leq t_0<t_1<\cdots$ are constants, $g_a(\cdot)$ is a measurable function and in addition
\begin{equation}
\sum_{i=0}^{\infty}|g_i| + \int_0^{\infty}|g_a(t)|\,dt < \infty.
\end{equation}
\end{definition}




For any $\boldsymbol{G}\in\mathcal{S}$, its associated transfer function, i.e., the Laplace transform of $g$, will be denoted by $G$, and its region of convergence includes the closed right-half plane. Moreover, its Fourier transform, which can be computed by taking $g(t)=0$ for all $t<0$, corresponds to its Laplace transform at the imaginary axis. Where appropriate we will consider either its Laplace transform (${G}:\bar{\mathbb{C}}_{+}\rightarrow \mathbb{C}$, $s\mapsto {G}(s)$ where $\bar{\mathbb{C}}_{+} = \{ s\in\mathbb{C}: Re(s)\geq 0\}$), or its Fourier Transform  (${G}: j \mathbb{R}\rightarrow \mathbb{C}$, $j\omega\mapsto {G}(j\omega))$).

 \begin{remark}
 Vidyasagar \cite{vidyasagar93} uses the notation $\mathcal{A}$ for the class of operators we denote $\mathcal{S}$ (and $\hat{\mathcal{A}}$ for the corresponding class of transfer functions; the class of operators is denoted $\mathcal{A}(0)$ in \cite{Curtain91}).
  Although the corresponding class of transfer functions is less general than $\mathcal{H}_\infty$, from an operator point-of-view the generality is appropriate \cite{vidyasagar93}.
 Certainly the class is considerably more general than the class of operators with transfer functions in $\mathcal{R}\mathcal{H}_{\infty}$. 
 Our notation is chosen  to avoid confusion with the Altshuller multipliers (defined below). Theorems~\ref{thm:ss} and~\ref{thm:dd} below are restricted to a subclass of $\mathcal{S}$. 
 \end{remark}



\subsection{Lurye systems}\label{sec:Lurye}

We are concerned with the behaviour of the closed-loop system depicted in Fig.~\ref{fig:Lurye} and defined as follows:
\begin{definition}\label{def:lurye}
A {\bf Lurye system} is a closed-loop system, assumed to be well-posed, with dynamics
\begin{equation}
y_1=\boldsymbol{G}u_1,\mbox{ } y_2=\boldsymbol{\phi} (u_2),\mbox{ } u_1=r_1-y_2 \mbox{ and }u_2 = y_1+r_2,\label{eq:Lurye}
\end{equation}
with $\boldsymbol{G}\in\mathcal{S}$ and $\boldsymbol{\phi}\in\Phi^m$.
We will denote by $\boldsymbol{L}^{y_j}_{r_i}$ and $\boldsymbol{L}^{u_j}_{r_i}$ the closed-loop maps from $r_i$ to $y_j$ and to $u_j$ respectively.
\end{definition}

\begin{remark}\label{rem:ic2}
Non-zero initial conditions can be accommodated in our definition of finite gain stability for such a Lurye system (c.f. Remark~\ref{rem:ic1}). Specifically, since the nonlinearity $\boldsymbol{\phi}$ is Lipschitz and the LTI transfer function $G$ admits a minimal state-space representation,  non-zero initial conditions can be accommodated by extending the time line backwards and including some fictitious exogenous signal over this extension. See \cite{vidyasagar93}, pp 290-291.
\end{remark}

 The Lurye system is stable if $r_1,r_2\in\mathcal{L}_2$  implies $u_1,u_2,y_1,y_2\in\mathcal{L}_2$. It is FGS if there is some $h<\infty$ such that
\begin{equation}
    \begin{split}
\|y_i\|\leq h (\|r_1\|+\|r_2\|)\text{ and }\|u_i\|\leq h (\|r_1\|+\|r_2\|),\\ 
\text{for }i=1,2\text{ and for all }r_1, r_2 \in \mathcal{L}_2.
    \end{split}
\end{equation}
Since $\boldsymbol{G}$ is LTI and stable, if $\boldsymbol{L}^{y_2}_{r_2}$ is FGS then all other closed-loop maps are FGS. A similar statement is true if $\boldsymbol{L}^{y_2}_{r_2}$ is FGOS. 

We will consider only time-invariant or periodic nonlinearities and correspondingly state specifically whether $\phi\in\Phi^{ti}$ (characterised by some $Q\in\mathcal{Q}$) or $\phi\in\Phi^p_T$ for some $T$ (characterised by some $N\in\mathcal{N}_T^p$). We will also specify, when appropriate, if $\phi\in\Phi^{sr}_k$ for some $k$ and if $\phi\in\Phi^{odd}$.  We will also, where appropriate, specify whether $r_2$ is zero, in $\mathcal{L}_2$ or in $\mathcal{P}$.

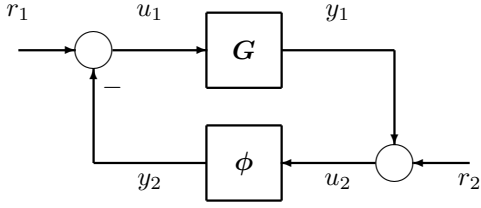
\begin{figure}[tbp]
\begin{center}
\input{lure_v4}
\end{center}
\caption{Lurye system.} 
\label{fig:Lurye}
\end{figure}



\subsection{Multiplier theory}\label{subsec:mult}


\begin{definition}[\cite{Zames68}]\label{def2a}
Let $\mathcal{M}$ be the class of continuous-time convolution operators $\boldsymbol{M}:\mathcal{L}_2\rightarrow\mathcal{L}_2$ whose (possibly non-causal) impulse response is given by 
\begin{equation}\label{m_def}
m(t)=\delta(t)-h(t)-\sum_{i=1}^{\infty}h_i\delta(t-t_i),
\end{equation}
with $h(t)\geq 0$ for all $t\in\mathbb{R}$, $h_i\geq 0$ and $t_i\neq 0$ for all $i$ and
\begin{equation}\label{h_ineq}
\int_{-\infty}^{\infty}h(t)\,dt + \sum_{i=1}^{\infty}h_i < 1.
\end{equation}
We say $\boldsymbol{M}$ is an {\bf OZF multiplier} if $\boldsymbol{M}\in\mathcal{M}$.
\end{definition}

\begin{remark}
It is possible to broaden the class of OZF multipliers by replacing (\ref{m_def}) with
\begin{equation}
m(t)=m_0\left (\delta(t)-h(t)-\sum_{i=1}^{\infty}h_i\delta(t-t_i)\right ),
\end{equation}
with $m_0>0$, but we can set $m_0=1$ without loss of generality. 
\end{remark}

\begin{remark}\label{rem:strict}
    Since classical loop transformation techniques \cite{Zames68, desoer75} require $\boldsymbol{M}^{-1}$ to exist, we impose a strict inequality in (\ref{h_ineq}).  If instead one uses  homotopy arguments \cite{Megretski97}  the requirement is removed, so a non-strict inequality can be used. This relaxation is sometimes useful for parametrizing the class of multipliers, while ultimately leading to the same results. We used a non-strict inequality in our preliminary results \cite{Heath:24}. The distinction is discussed in \cite{Carrasco12}. 
\end{remark}

\begin{definition}[\cite{Zames68}]\label{def2b}
Let $\mathcal{M}_{\text{odd}}$ be the class of  continuous-time convolution operators $\boldsymbol{M}:\mathcal{L}_2\rightarrow\mathcal{L}_2$ whose (possibly non-causal) impulse response is given by (\ref{m_def})
with 
\begin{equation}\label{h_ineq_odd}
\int_{-\infty}^{\infty}|h(t)|\,dt + \sum_{i=1}^{\infty}|h_i| < 1.
\end{equation}
We say $\boldsymbol{M}$ is an {\bf OZF multiplier for odd nonlinearities} if $\boldsymbol{M}\in\mathcal{M}_{\text{odd}}$.
\end{definition}

\begin{definition}\label{def1}
Let 
${M}:j \mathbb{R}\rightarrow\mathbb{C}$ and let ${G}:j \mathbb{R}\rightarrow\mathbb{C}$. 
We say ${M}$ is {\bf suitable} for ${G}$ if there exists $\varepsilon>0$ such that
\begin{align}\label{suitable}
\mbox{Re}\left \{
				{M}(j\omega) {G}(j\omega)
			\right \} > \varepsilon\mbox{ for all } \omega \in \mathbb{R}.
\end{align}
We also say a linear operator $\boldsymbol{M}$ is suitable for $\boldsymbol{G}$ if their respective frequency responses satisfy (\ref{suitable}).
\end{definition}

\begin{theorem}[\cite{Zames68,desoer75}]
    A Lurye system (Definition~\ref{def:lurye}) with  $\boldsymbol{\phi}\in\Phi^{ti}$   is FGS if there is 
an $\boldsymbol{M}\in\mathcal{M}$ 
suitable for~$\boldsymbol{G}$. A Lurye system with  $\boldsymbol{\phi}\in\Phi_k^{sr}\cap\Phi^{ti}$  is FGS if there is 
an $\boldsymbol{M}\in\mathcal{M}$  suitable for~$1/k+\boldsymbol{G}$. 

If, in addition, $\phi\in\Phi^{odd}$ then the same statements can be made for $\boldsymbol{M}\in\mathcal{M}_{odd}$.
\end{theorem}

These results can be expressed (and derived) in terms of Integral Quadratic Constraints (IQCs) \cite{Megretski97} (c.f. Remark~\ref{rem:strict}). This is particularly useful for calculating bounds on the $\mathcal{L}_2$ gain \cite{Turner12,Veenman14,Bertolin25}. If $y_2=\boldsymbol{\phi}(u_2)$ with $\boldsymbol{\phi}\in\Phi^{sr}_k\cap\Phi^{ti}$ and $u_2,y_2\in\mathcal{L}_2$ then their Fourier transforms satisfy 
\begin{equation}
\int_{-\infty}^\infty\left [
\begin{array}{c}
\hat{u}_2(j\omega)\\
\hat{y}_2(j\omega)
\end{array}
\right ]^*
\Pi(j\omega)
\left [
\begin{array}{c}
\hat{u}_2(j\omega)\\
\hat{y}_2(j\omega)
\end{array}
\right ]
\,d\omega \geq 0,
\end{equation}
with
\begin{equation}
\Pi(j\omega)=\left [
    \begin{array}{cc}
        0 & M^*(j\omega)\\
        M(j\omega) & -(M(j\omega)+M^*(j\omega))/k
    \end{array}
\right ],
\end{equation}
and $M(j\omega)$ is the frequency response of some $\boldsymbol{M}\in\mathcal{M}$.

We write $G$ for $G(j\omega)$ and $M$ for $M(j\omega)$. It follows (see Fig.~\ref{fig:Lurye_L2one}) that we can bound  the  $\mathcal{L}_2$ gain from $r_2$ to $y_2$ of 
a Lurye system with $\boldsymbol{\phi}\in\Phi^{sr}_k\cap\Phi^{ti}$  by the smallest $h$ satisfying
\begin{eqnarray}
\left [
    \begin{array}{cc}
    -G & 1\\
    1 & 0\\
    1 & 0\\
    0 & 1
    \end{array}
\right ]^*
\left [
    \begin{array}{cccc}
        0 & 0 & M^* & 0\\
        0 & 1/h & 0 & 0\\
        M & 0 & -(M+M^*)/k & 0\\
        0 & 0 & 0 & -h
    \end{array}
\right ]\nonumber\\
\left [
    \begin{array}{cc}
    -G & 1\\
    1 & 0\\
    1 & 0\\
    0 & 1
    \end{array}
\right ]
< -\epsilon I,\text{ for all $\omega$, for some $\epsilon>0$.}
\end{eqnarray}
This reduces to
\begin{equation}
\left [
    \begin{array}{cc}
    1/h - G^*M^*-MG-(M+M^*)/k & M\\
    M^* & -h
    \end{array}
\right ]
<-\epsilon I,
\end{equation}
and hence, taking limits, we can bound the gain by
\begin{equation}\label{first_gamma}
    h_M = \sup_{\omega} \frac{1+M^*M}{G^*M^*+MG+(M+M^*)/k}.
\end{equation}
A similar expression  when the nonlinearity is monotonic but not slope-restricted follows by letting $k\rightarrow\infty$.

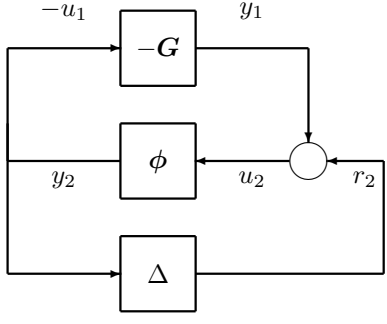
\begin{figure}[tbp]
\begin{center}
\input{Lure_delta1}
\end{center}
\caption{Block diagram for bound on $\mathcal{L}_2$ gain from $r_2$ to $y_2$.} 
\label{fig:Lurye_L2one}
\end{figure}

Similarly 
an upper bound on the gain from $r_2$ to $u_2$ satisfies
\begin{eqnarray}
\left [
    \begin{array}{cc}
    -G & 1\\
    -G & 1\\
    1 & 0\\
    0 & 1
    \end{array}
\right ]^*
\left [
    \begin{array}{cccc}
        0 & 0 & M^* & 0\\
        0 & 1/h & 0 & 0\\
        M & 0 & -(M+M^*)/k & 0\\
        0 & 0 & 0 & -h
    \end{array}
\right ]\nonumber\\
\left [
    \begin{array}{cc}
    -G & 1\\
    -G & 1\\
    1 & 0\\
    0 & 1
    \end{array}
\right ]
<-\epsilon I,
\end{eqnarray}
which reduces to 
\[  
 \left [
    \begin{array}{cc}
    G^*G/h - G^*M^*-MG-(M+M^*)/k & M-G^*/h\\
    M^*-G/h & 1/h-h
    \end{array} 
\right ]
\]
\begin{equation}
  <-\epsilon I.
\end{equation}
Hence we require $h> 1$ and
\begin{equation}\label{qi}
2h^2\text{Re}[M(1/k+G)]  -h (G^*G+MM^*)  -2\text{Re}[M/k]
  \geq 0.
\end{equation}
This gives the bound
\[h_m = \sup_{\omega}r(j\omega)\]
where, at each frequency, $r(j\omega)$ is the positive root of (\ref{qi}) satisfied with equality. If the nonlinearity is monotone but not slope-restricted this reduces to
\begin{equation}
h_m  =  \sup_{\omega}\frac{G^*G+MM^*}{MG+G^*M^*}.
\end{equation}



We can find bounds on the gains from $r_1$ or $r_2$ to $u_1$, $u_2$, $y_1$ or $y_2$ in a similar fashion. Results are summarised in Table~\ref{tab1}.

\begin{table*}
    \centering
    \begin{tabular}{|c|c|c|}
    \hline & $r_1$ & $r_2$\\
    \hline $u_1$ & 
    $
            \begin{array}{c}
                {\displaystyle{\sup_{\omega} h}} \text{ such that } h\geq 1 \text{ and }
                    \\
                h^22\text{Re}[M(G+1/k)] - h(1+|MG|^2)-2\text{Re}[M/k]=0
            \end{array}
    $
   & $\displaystyle{\sup_\omega\frac{1+|M|^2}{2\text{Re}[M(G+1/k)]}}$ \\
    \hline $u_2$ &
    $\begin{array}{c}
    \displaystyle{\sup_\omega h} \text{ such that } \displaystyle{h\geq \sup_\omega |G|^2}\text{ and }\\
    \displaystyle{h^22\text{Re}[M(G+1/k)] - h|G|^2(1+|M|^2)-|G|^22\text{Re}[M/k]=0}
    \end{array}
    $
    &
    $\begin{array}{c}
    \displaystyle{\sup_\omega h}\text{ such that } h\geq 1 \text{ and}\\
    \displaystyle{h^22\text{Re}[M(G+1/k)] - h(|G|^2+|M|^2)-\text{Re}[M/k]=0}
    \end{array}
    $\\
    \hline 
    $y_1$ & 
    $
        \begin{array}{c}
        \displaystyle{\sup_\omega h} \text{ such that } h\geq \displaystyle{\sup_\omega|G|^2} \text{ and }\\
        h^22\text{Re}[M(G+1/k)] - h|G|^2(1+|M|^2)-|G|^22\text{Re}[M/k]=0
        \end{array}
    $
        & $\displaystyle{\sup_\omega\frac{|G|^2+|M|^2}{2\text{Re}[M(G+1/k)]}}$ \\
    \hline $y_2$ & $\displaystyle{\sup_\omega\frac{1+|MG|^2}{2\text{Re}[M(G+1/k)]}}$ & $\displaystyle{\sup_\omega\frac{1+|M|^2}{2\text{Re}[M(G+1/k)]}}$ \\
    \hline
    \end{tabular}
    \caption{Bounds on the $\mathcal{L}_2$ gains for Lurye systems.}\label{tab1}
\end{table*}


\begin{definition}\label{def:alt}
Let $\mathcal{A}_T\subset\mathcal{M}$ be the class of continuous-time convolution operators whose impulse response is given by 
\begin{equation}\label{alt_imp}
m(t)=\delta(t)-\sum_{n=-\infty}^{\infty}h_n \delta(t-nT),
\end{equation}
with 
\begin{equation}\label{alt_cond}
    h_0=0\text{, }h_n\geq 0\text{ for all } n  \text{ and }
\sum_{n=-\infty}^{\infty}h_n < 1.
\end{equation}
We say $\boldsymbol{M}$ is an {\bf Altshuller multiplier with period} $T$ if $\boldsymbol{M}\in\mathcal{A}_T$.
\end{definition}

\begin{remark}
Altshuller \cite{Altshuller:11} uses the frequency response to define multipliers equivalent to those in Definition~\ref{def:alt}.
\end{remark}

\subsection{Discrete-time systems}

Although our development is for continuous-time systems, most definitions and results carry over to discrete-time,  with 
the spaces $\ell$ of real-valued sequences $h:\mathbb{Z}^+\rightarrow \mathbb{R}$ and $\ell_2$ of square-summable sequences $h:\mathbb{Z}^+\rightarrow \mathbb{R}$  taking the places of $\mathcal{L}_{2e}$ and   $\mathcal{L}_2$ respectively. Similarly $\ell_2$ gain takes the place of $\mathcal{L}_2$ gain.

We require some of these to be stated explicitly. 
The following is the counterpart to Definition~\ref{def2a}.
\begin{definition}\label{def2a_d}
Let $\mathcal{M}^d$ be the class of discrete-time convolution operators $\boldsymbol{M}:\ell_2\rightarrow\ell_2$ whose (possibly non-causal) impulse response is given by 
\begin{equation}\label{def_md}
m(n) = \delta(n)-\sum_{k=-\infty}^{\infty}h_k \delta(n-k),
\end{equation}
for $n\in\mathbb{Z}$ with $h_0$=0, $h_k\geq 0$ for all $k$ and
\begin{equation}
\sum_{k=-\infty}^{\infty}h_k < 1.
\end{equation}
We say $\boldsymbol{M}$ is a {\bf discrete-time OZF multiplier} if $\boldsymbol{M}\in\mathcal{M}^d$.
\end{definition}

\begin{remark}
The original definition of the discrete-time multipliers \cite{Willems68} includes linear time-varying multipliers. It is sufficient to consider  LTI multipliers \cite{Kharitenko23,Su23}.
\end{remark}
The class of discrete-time OZF multipliers for odd nonlinearities is defined analogously. Let $\mathbb{D}$ denote the unit circle. Then we have the following counterpart to Definition~\ref{def1}.
\begin{definition}\label{def1d}
Let 
${M}:\mathbb{D}\rightarrow\mathbb{R}$ and let ${G}:\mathbb{D}\rightarrow\mathbb{C}$. 
We say ${M}$ is {\bf suitable} for ${G}$ if
\begin{align}
\mbox{Re}\left \{
				{M}(e^{j\omega}) {G}(e^{j\omega})
			\right \} > 0\mbox{ for all } \omega \in [0,2\pi).
\end{align}
\end{definition}
The counterpart to Theorem 1 is direct and standard \cite{Willems68,Willems71}. 
Finally we define the discrete-time counterpart to the class $\mathcal{A}_T$.
\begin{definition}\label{def:altd}
Let $\mathcal{A}^d_N\subset\mathcal{M}^d$ with $N\in\mathbb{Z}^+$ be the class of discrete-time convolution operators whose impulse response is given by 
\begin{equation}
m(n)=\delta(n)-\sum_{k=-\infty}^{\infty}h_k \delta(n-kN),
\end{equation}
with 
\begin{equation}
    h_0=0\text{, }h_k\geq 0\text{ for all } k  \text{ and }
\sum_{k=-\infty}^{\infty}h_k < 1.
\end{equation}
We say $\boldsymbol{M}$ is a {\bf discrete-time Altshuller multiplier with period}~$N$ (as a multiple of sampling period) if $\boldsymbol{M}\in\mathcal{A}^d_N$.
\end{definition}


\subsection{Multivariable systems}

Similarly, while our development is for single-input single-output systems, many results carry over in a straightforward manner to multivariable systems. In particular the development of \cite{Altshuller:11,Altshuller:13}, which sets the framework for Chapter~\ref{sec:periodic}, is for multivariable systems. 

If the nonlinearity $\boldsymbol{\phi}$ is characterised by $N:\mathbb{R^+}\times\mathbb{R}^n\rightarrow\mathbb{R}^n$ then we have the additional requirement that $\int_A^BN(t,x)\,dx$ be independent of path. This is equivalent to saying there is some $P:\mathbb{R^+}\times\mathbb{R}^n\rightarrow\mathbb{R}$ that is convex in the second variable with $N(t,\cdot)=\nabla P(t,\cdot)$ for all $t$. If $N$ is sufficiently smooth this is in turn equivalent to the requirement that
\begin{equation}
    \left [\frac{\partial N(t,x)}{\partial x_i}\right ]_j = \left [\frac{\partial N(t,x)}{\partial x_j}\right ]_i\text{ for all }i,j.
\end{equation}
See \cite{Willems72,Safonov2000}.

\section{Power analysis}

In this section we consider FGS systems. The application to Lurye systems where finite-gain stablilty is guaranteed by the existence of a suitable OZF multiplier is immediate.

Zames \cite{Zames66a} argues that ``in order to behave properly'' a system's ``outputs must not be critically sensitive to small changes in inputs - changes such as those caused by noise.'' Further he argues that the input-output map must be continuous to ensure it is ``not critically sensitive to noise.'' Here we show that if the noise is a power signal then finite-gain stability suffices.

The following is standard for linear systems \cite{Zhou96} but, to the authors' knowledge and with the exception of \cite{Heath:24}, has not been previously stated for nonlinear systems.

\begin{theorem}\label{thm:power}
Suppose $u\in\mathcal{P}$ and $y=\boldsymbol{H}(u)$ where $\boldsymbol{H}$ is FGS with gain $h$.  Then $y\in\mathcal{P}$ with $\|y\|_P\leq h \|u\|_P$.
\end{theorem}
\begin{proof}We find
\begin{align}
    \|y\|_P^2 & = \limsup_{T\rightarrow \infty} \frac{1}{T}\|(\boldsymbol{H}(u))_T\|^2, \nonumber \\
     & = \limsup_{T\rightarrow \infty} \frac{1}{T}\|(\boldsymbol{H}(u_T))_T\|^2\text{ since $\boldsymbol{H}$ is causal,}\nonumber \\
     & \leq \limsup_{T\rightarrow \infty} \frac{1}{T}\|\boldsymbol{H}(u_T)\|^2,\nonumber\\
     & \leq \limsup_{T\rightarrow \infty} \frac{1}{T} h^2\|u_T\|^2,\nonumber\\
     & = h^2 \|u\|_P^2. 
\end{align}
\end{proof}
Hence if we define the power gain of $\boldsymbol{H}$ to be
\begin{equation}
    h_P = \sup_{u\in\mathcal{P},\|u\|_P>0} \frac{\|\boldsymbol{H}(u)\|_P}{\|u\|_P},
\end{equation}
then $h_P\leq h$.

Suppose $u_1,u_2\in\mathcal{P}$ and $\boldsymbol{H}$ is FGS with gain $h$. Since $\|\cdot\|_P$ is a seminorm we have the triangle inequalities
\begin{align}
        \|\boldsymbol{H}(u_1) & \pm  \boldsymbol{H}(u_2)\|_P^2 \nonumber\\
          & \leq \|\boldsymbol{H}(u_1)\|_P^2+\|\boldsymbol{H}(u_2)\|_P^2 + 2 \|\boldsymbol{H}(u_1)\|_P\|\boldsymbol{H}(u_2)\|_P\nonumber\\
         & \leq h^2 \left [\|u_1\|_P^2+\|u_2\|_P^2+2\|u_1\|_P\|u_2\|_P\right ],
\end{align}
and
\begin{align}\label{ineq:power}
        \|\boldsymbol{H}(u_1+& u_2)\|_P^2 \leq h^2 \|u_1+u_2\|^2 \nonumber\\
        & \leq h^2 \left [\|u_1\|_P^2+\|u_2\|_P^2+2\|u_1\|_P\|u_2\|_P\right ].
\end{align}
In particular we may say:
\begin{theorem}\label{thm:power_L2}
Suppose $\boldsymbol{H}:\mathcal{L}_{2e}\rightarrow\mathcal{L}_{2e}$ is FGS with gain $h$. Suppose further that $u_1\in\mathcal{L}_2$ and $u_2\in\mathcal{P}$.  Then
\begin{equation}
        \|\boldsymbol{H}(u_1 + u_2)\|_P  \leq h \|u_2\|_P.
\end{equation} 
\end{theorem}
\begin{proof}
The result follows immediately from (\ref{ineq:power}) since 
 $\|u_1\|_P=0$.
\end{proof}
Hence, if we add small power noise to an $\mathcal{L}_2$ input signal then the output power is small provided the system is FGS, irrespective of the continuity of the input-output map.

\section{Bias signal analysis}\label{sec:bias}
Theorem~\ref{thm:power_L2} has an immediate counterpart when $u_1$ is a bias signal and $\boldsymbol{H}$ is FGOS.
\begin{theorem}\label{thm:power_offset}
Suppose $\boldsymbol{H}:\mathcal{L}_{2e}\rightarrow\mathcal{L}_{2e}$ is FGOS with steady state map $H_0$ and offset gain $h$. Suppose further that $u_1$ is a bias signal with bias $\bar{u}_1$ and  $u_2\in\mathcal{P}$.  Then
\begin{equation}
 \|\boldsymbol{H}(u_1+u_2)-H_0(\bar{u}_1)\|_P  \leq h\|u_2\|_P.
 \end{equation}
\end{theorem}
\begin{proof}
The result follows from Theorem~\ref{thm:power_L2} 
since
\begin{equation}
 \|\boldsymbol{H}(u_1+u_2)-H_0(\bar{u}_1)\|_P = \|\boldsymbol{H_{\bar{u}}}(u_1-\bar{u}_1\mathds{1}+u_2)\|_P,
 \end{equation}
 where $\boldsymbol{H_{\bar{u}}}$ is given by Definition~\ref{def:offset}. 
\end{proof}
Hence, if we add small power noise to a input bias signal then the output power (measured about the noise-free output bias) is small provided the system is FGOS, irrespective of the continuity of the input-output map.

The application to a Lurye system with time-invariant nonlinearity is immediate provided we can show the system is FGOS. It turns out that the existence of a suitable OZF multiplier $\boldsymbol{M}\in\mathcal{M}$ guarantees this, but the existence of a suitable $\boldsymbol{M}\in\mathcal{M}_{\text{odd}}$ does not.

\begin{theorem}\label{thm:offset}
    A Lurye system (Definition~\ref{def:lurye}) with  $\boldsymbol{\phi}\in\Phi^{ti}$ is FGOS if there is 
an $\boldsymbol{M}\in\mathcal{M}$ suitable for~$\boldsymbol{G}$. A Lurye system with  $\boldsymbol{\phi}\in\Phi^{sr}_k\cap\Phi^{ti}$ is FGOS if there is 
an $\boldsymbol{M}\in\mathcal{M}$ suitable for~$1/k+\boldsymbol{G}$. 
\end{theorem}

\begin{proof}
It is sufficient to show $\boldsymbol{L}^{y_2}_{r_2}$ is FGOS. Let $r_1=0$ without loss of generality and let $r_2$ be a bias signal with bias $\bar{r}_2$. 

Suppose first that $r_2$ is constant, i.e. 
$r_2=\bar{r}_2\mathds{1}$ for some $\bar{r}_2\in\mathbb{R}$
and $\boldsymbol{G}$ is a fixed gain $g\in \mathbb{R}$. 
The monotonicity of $\boldsymbol{\phi}$ ensures there is a unique fixed solution 
$y_2=\bar{y}_2\mathds{1}$ \cite{desoer75}.
This defines our candidate steady state map $H_o(\bar{r}_2)=\bar{y}_2$ from $r_2$ to $y_2$. 
The input to the nonlinearity is $u_2=\bar{u}_2\mathds{1}$ where $\bar{u}_2=\bar{r}_2-g\bar{y}_2$.
If the nonlinearity is characterised by~$Q$ then $\bar{y}_2=Q(\bar{u}_2)$.

Now consider our original system where $\boldsymbol{G}$ has frequency response $G(j\omega)$ and set $g=G(0)$. We can define a normalized Lurye system with the same linear element $\boldsymbol{G}$ but nonlinear element $\boldsymbol{\phi}_{\bar{r}_2}\in\Phi^{ti}$
 characterised by $Q_{\bar{r}_2}$ where
\begin{equation}
Q_{\bar{r}_2}(x) = Q (x+\bar{u}_2) - \bar{y}_2 \text{ for all }x\in\mathbb{R}.
\end{equation}
Denote $\boldsymbol{H_{\bar{r}_2}}$ as the map from $r_2-\bar{r}_2\mathds{1}$ to $y_1-\bar{y}_1\mathds{1}$. If $\boldsymbol{M}\in\mathcal{M}$ is suitable for $\boldsymbol{G}$ then $\boldsymbol{H_{\bar{r}_2}}$ is FGS. It follows that $\boldsymbol{L}^{y_2}_{r_2}$ for the original system is FGOS.

If in addition $\boldsymbol{\phi}\in\Phi^{sr}_k$ then $\boldsymbol{\phi}_{\bar{r}_2}\in\Phi^{sr}_k$ also.
\end{proof}

\begin{remark}
There is no corresponding result when $\boldsymbol{M}\in\mathcal{M}_{\text{odd}}-\mathcal{M}$. In this case we can have  $\boldsymbol{\phi}_{\bar{r}_2}\notin\Phi^{odd}$ even if $\boldsymbol{\phi}\in\Phi^{odd}$. Hence if $\boldsymbol{M}\in\mathcal{M}_{\text{odd}}-\mathcal{M}$ is suitable for $\boldsymbol{G}$ there is no guarantee that $\boldsymbol{H_{\bar{r}_2}}$ is stable.
\end{remark}

\begin{remark}
Suppose $\boldsymbol{M}$ is suitable for $\boldsymbol{G}$ (or for $1/k+\boldsymbol{G}$) and can be used to ensure the $\mathcal{L}_2$ gain of $\boldsymbol{L}^{y_2}_{r_2}$ is bounded above by $h$. It follows from the proof of Theorem~\ref{thm:offset} that the offset gain of $\boldsymbol{L}^{y_2}_{r_2}$ is also bounded above by $h$.
\end{remark}

\section{Periodic excitation}\label{sec:periodic}

\subsection{Lurye systems with periodic nonlinearities}\label{sec:pernon}

Altshuller's preliminary results in \cite{Altshuller:11,Altshuller:13} concern Lurye systems where $\boldsymbol{\phi}\in\Phi^{sr}_k\cap\Phi^p_T$ for some period $T$ and slope $k$. 
In addition there are regularity conditions on both the linear part (described by integral equations) and the nonlinearity ($N(\cdot,\cdot)$ is required to continuous with respect to its first variable and differentiable with respect to its second).
Altshuller establishes that the Altshuller multipliers are sufficient for absolute $\mathcal{L}_2$ stability, using delay-IQCs to establish the results. Here we show that classical methods can be used to obtain absolute $\mathcal{L}_2$ stability for very general LTI systems and a more general class of periodic nonlinearity.


We begin by showing that the Altshuller multipliers are necessary and sufficient to preserve the positivity of periodic nonlinearities.

\begin{lemma}
    Let $\boldsymbol{\phi}\in\Phi^p_T$ for some $T>0$. If $u\in\mathcal{L}_2$ and $y=\boldsymbol{\phi}(u)$ then 
\begin{equation}
    \int_{-\infty}^{\infty}
u(t+nT)y(t)\,dt \leq \int_{-\infty}^{\infty}
u(t)y(t)\,dt,
\end{equation}
for all $n\in\mathbb{Z}$.
\end{lemma}
\begin{proof}[Proof (c.f. \cite{desoer75}, p205)]
$\mbox{ }$\\
Let $\mathcal{\phi}$ be characterised by 
$N\in\mathcal{N}^{p}_T$.
Define $F(t,p)$ so that 
\begin{equation}
    \frac{\partial F(t,p)}{\partial p}= N(t,p)\text{ and }F(t,0)=0\text{ for all }t.
\end{equation}
Then $F$ is a periodic function in its first variable and a convex function in its second. In addition,
it follows that if $u\in\mathcal{L}_2$ then the signal $f\in\mathcal{L}_1$ where $f(t)=F(t,u(t))$.
For any $a,b\in\mathbb{R}$ we have
\begin{align}
    F(t,a)-F(t,b) & =\int_b^aN(t,p)\,dp,\nonumber\\
    & \leq (a-b)N(t,a).
\end{align}
Put  $a=u(t)$ and $b=u(t+nT)$. So
\begin{equation}\label{ineq}
    F(t,u(t))-F(t,u(t+nT))
    \leq (u(t)-u(t+nT))y(t).
\end{equation}
Integrating the left hand side of (\ref{ineq}) yields
\begin{align}
    \int_{-\infty}^{\infty} & F(t,u(t)) -F(t,u(t+nT))\, dt \nonumber \\
      & = \int_{-\infty}^{\infty} F(t,u(t))-F(t+nT,u(t+nT))\, dt, \nonumber \\
      &  = 0.
\end{align}
Hence integrating the right hand side of (\ref{ineq}) yields the result.
%
%
%
\end{proof}

\begin{lemma}$\mbox{ }$\\
        (a) Let 
        $\boldsymbol{\phi}\in\Phi^p_T$ for some $T>0$. Let $\boldsymbol{M}\in\mathcal{A}_T$. Then for all $u\in\mathcal{L}_2$ we have
\begin{equation}
\int_{-\infty}^{\infty} (\boldsymbol{M}u)(t)\,(\boldsymbol{\phi}( u))(t)\,dt \geq 0.
\end{equation}
        (b) Let $\boldsymbol{M}\in\mathcal{M}-\mathcal{A}_T$ for some $T>0$. Then there exists a 
        $\boldsymbol{\phi}\in\Phi^p_T$ and a  $u\in\mathcal{L}_2$ such that
\begin{equation}
\int_{-\infty}^{\infty} (\boldsymbol{M}u)(t)\,(\boldsymbol{\phi}( u))(t)\,dt < 0.
\end{equation}
\end{lemma}
\begin{proof} $\mbox{ }$\\
(a)
    Let $\boldsymbol{M}$ have impulse response $m(t)$ given by (\ref{alt_imp}) and satisfying (\ref{alt_cond}). 
Then
\begin{align}
\int_{-\infty}^{\infty} & (\boldsymbol{M}u)(t)  (\boldsymbol{\phi}( u))(t)\,dt
 = \int_{-\infty}^{\infty} (m*u)(t) y(t)\,dt,\nonumber\\
& = \int_{-\infty}^{\infty} u(t)y(t)\,dt - \sum_{n=-\infty}^{\infty}h_n \int_{-\infty}^{\infty} u(t-nT)y(t)\,dt,\nonumber\\
& \geq 0.
\end{align}
(b) We construct such a nonlinearity $\boldsymbol{\phi}$ and signal $u$ as follows.
Let $\boldsymbol{\phi}$ be characterised by 
$N\in\mathcal{N}^p_T$
with
\begin{equation}
N(t,u(t)) = \left \{\begin{array}{ll}(1+\Delta)u(t) & \text{for }0\leq t<T/2,\\
u(t)/(1+\Delta) & \text{for }T/2\leq t<T.\end{array}\right .
\end{equation}
Let $\bar{u}$ be periodic with
\begin{equation}
\bar{u}(t) = \left \{\begin{array}{ll}1  & \text{for }0\leq t<T/2,\\
1+\Delta & \text{for }T/2\leq t<T.\end{array}\right .
\end{equation}
Then the output $\bar{y}=\boldsymbol{\phi}(\bar{u})$ is also periodic with 
\begin{equation}
\bar{y}(t) = \left \{\begin{array}{ll}1+\Delta  & \text{for }0\leq t<T/2,\\
1 & \text{for }T/2\leq t<T.\end{array}\right .
\end{equation}
Then
\begin{align}
\int_0^T \bar{u}(t+nT)\bar{y}(t) \, dt & =  T (1 + \Delta  )\text{ for any }n\in\mathbb{Z}.
\end{align}
Suppose $0<\tau < T$. 
We find
\begin{equation}
\int_{0}^{T}  \bar{u}(t+\tau)  \bar{y}(t)\, dt
=\left \{
    \begin{array}{l} T (1 + \Delta  )+\tau \Delta^2 \\
    \hspace{1.5cm}\text{ for }0<\tau\leq T/2,\\
    T(1+\Delta  )+(T-\tau)\Delta^2\\
     \hspace{1.5cm}\text{ for }T/2<\tau < T.
    \end{array}
    \right .
%
\end{equation}

Now suppose $\boldsymbol{M}\in \mathcal{M}-\mathcal{A}_T$. 
We can write its impulse response as in Definition~\ref{def2a} where either $h_i>0$ for some $t_i$ that is not an integer multiple of $T$ or there exists some $\underline{h}>0$ such that $h(t)\geq\underline{h}$ almost everywhere on some interval $(t_j,t_j+\delta)$ which contains no integer multiple of $T/2$.


Suppose first there is some $t_i$ with $h_i>0$. 
Write $t_i = n_i T + \tau_i$
with $0<\tau_i<T$.
Then 
\begin{align}
\int_0^T \bar{y}(t) [m*\bar{u}](t)\, dt & 
< T(1+\Delta)-h_i\min(\tau_i,T-\tau_i)\Delta^2,\nonumber\\
& < 0 \text{ for }\Delta \text{ sufficiently big.}
\end{align}
Suppose instead there is some $\underline{h}>0$ with $h(t)\geq\underline{h}$ almost everywhere on some interval $(t_j,t_j+\delta)$ which contains no integer multiple of $T/2$. Write $t_j = n_j T + \tau_j$
with $0<\tau_j<T$. Then
\begin{align}
\int_0^T \bar{y}(t) [m*\bar{u}](t)\, dt & 
< T(1+\Delta)-\underline{h}\delta\min(\tau_j,T-\tau_j)\Delta^2,\nonumber\\
& < 0 \text{ for }\Delta \text{ sufficiently big.}
\end{align}

Finally we can truncate to find $u=\bar{u}_{\bar{T}}\in\mathcal{L}_2$ and $y=\bar{y}_{\bar{T}}\in\mathcal{L}_2$ with $\bar{T}$ sufficiently big that the inequalities still hold.
\end{proof}

We can now state a stability result that generalises Theorem~5 in \cite{Altshuller:11} for the single-input single-output case.

\begin{theorem}\label{thm:pernon}
    A Lurye system with  $\boldsymbol{\phi}\in\Phi^p_T$ for some $T>0$ is FGOS if there is 
an $\boldsymbol{M}\in\mathcal{A}_T$ suitable for~$\boldsymbol{G}$.  
%
%
If, in addition, $\boldsymbol{\phi}\in\Phi^{sr}_k$ for some $k>0$ then the Lurye system is FGOS if there is 
an $\boldsymbol{M}\in\mathcal{A}_T$ suitable for~$1/k+\boldsymbol{G}$.  
The closed-loop gains can be bounded by the expressions in Table~\ref{tab1} for any suitable $\boldsymbol{M}\in\mathcal{A}_T$.
\end{theorem}
\begin{proof}
Since $\mathcal{A}_T\subset \mathcal{M}$ any $\boldsymbol{M}\in\mathcal{A}_T$ inherits the properties of the OZF multipliers. In particular $\boldsymbol{M}^{-1}$ exists (see Remark~\ref{rem:strict} and \cite{Carrasco12}); this, together with Lemma~2a, establishes the positivity of the map from $x$ to $y$ where $u=\boldsymbol{M}^{-1}x$; similarly an appropriate factorization of $\boldsymbol{M}$ is guaranteed (see \cite{Carrasco12}); FGS then follows from standard multiplier theory (see, e.g., \cite{desoer75}). 
FGOS follows from the development of Chapter~\ref{sec:bias}. 
The expressions for 
closed-loop gains follow similar to the development of Section~\ref{subsec:mult}. 
\end{proof}



The generalisation of Theorem~\ref{thm:pernon} to either an $n\times n$ multivariable system or to a discrete-time system (or to both) is straightforward.



\subsection{Lurye systems with periodic excitation}

Altshuller \cite{Altshuller:11,Altshuller:13} uses results similar to those of Section~\ref{sec:pernon} to analyse Lurye systems with time-invariant nonlinearities but with periodic exogenous signals. The following generalises a  result in  \cite{Altshuller:11} for LTI systems whose dynamics are governed by integral equations.

Both the extension to more general LTI systems and the implications for bounds on the closed-loop gain follow immediately given our previous analysis.

\begin{theorem}\label{thm:perexc}

Consider a Lurye system where $\boldsymbol{\phi}\in\Phi^{ti}$. 
Let $r_2^*$ be periodic (and non-zero) with period $T>0$. Suppose when $r_1=0$ and $r_2=r_2^*$ there exists a (not-necessarily unique or attracting) non-zero periodic solution $y_2=y_2^*$ with period $T$. If there is 
an $\boldsymbol{M}\in\mathcal{A}_T$ suitable for~$\boldsymbol{G}$ then the periodic solution is unique and a global attractor. Further, let $h$ be the bound on the gain from $r_2$ to $y_2$ given in Table~\ref{tab1} with this multiplier. Then if $r_2-r_2^*\in\mathcal{L}_2$ we have $\|y_2-y_2^*\|\leq h \|r_2-r_2^*\|$ and if $r_2-r_2^*\in\mathcal{P}$ we have $\|y_2-y_2^*\|_P\leq h \|r_2-r_2^*\|_P$.

A similar result follows if, in addition, $\boldsymbol{\phi}\in\Phi^{sr}_k$ for some $k>0$ 
and there is an $\boldsymbol{M}\in\mathcal{A}_T$ suitable for~$1/k+\boldsymbol{G}$.

\end{theorem}
\begin{proof}
Let $\phi$ by characterised by $Q\in\mathcal{Q}$. Let $u_1^*=-y_2^*$, $y_1^*=\boldsymbol{G}u_1^*$ and $u_2^*=y_1^*+r_2^*$. Define the deviation variables $u^\partial_1 = u_1-u^*_1$, $u^\partial_2 = u_2-u^*_2$, $y^\partial_1 = y_1-y^*_1$, $y^\partial_2 = y_2-y^*_2$.
    Then 
    \begin{equation}
        u^\partial_1 = -y^\partial_1\text{, } u^\partial_2  =  y^\partial_2\text{, }  y^\partial_1 = \boldsymbol{G}u^\partial_1\text{ and }y^\partial_2=\boldsymbol{\phi}^\partial(u^\partial_2),
    \end{equation}
    where $\boldsymbol{\phi}^\partial$ is characterised by $N^\partial:\mathbb{R}^+\times\mathbb{R}\rightarrow\mathbb{R}$ with
    \begin{equation}
        N^\partial(t,u^\partial_2(t))=Q(u^\partial_2(t)+u^*_2(t))-Q(u_2^*(t)).
    \end{equation}
    See Fig~\ref{fig:Lurye_delta}.
    It follows that the periodic solution of the Lurye system with LTI element $\boldsymbol{G}$, nonlinearity $\boldsymbol{\phi}$ and periodic $r_2\neq 0$ is a global attractor if and only if the autonomous Lurye system with LTI element $\boldsymbol{G}$ and nonlinearity $\boldsymbol{\phi}^\partial$ is stable.
    Furthermore $\boldsymbol{\phi}^\partial\in\Phi^p_T$; 
    also $\boldsymbol{\phi}^\partial\in\Phi^{sr}_k$ whenever $\boldsymbol{\phi}\in\Phi^{sr}_k$.
    The existence of a suitable Altshuller multiplier ensures this system is FGS - in particular deviations tend to zero as $t\rightarrow\infty$.
\end{proof}

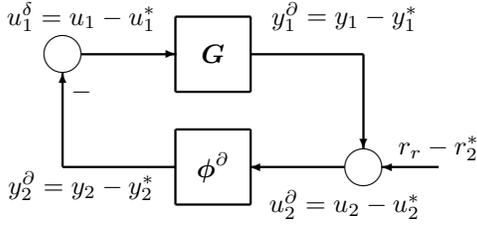
\begin{figure}[tbp]
\begin{center}
\input{lurye_delta}
\end{center}
\caption{Change of variables for proof of Theorem~\ref{thm:perexc}. $u_1^\partial$, $u_2^\partial$, $y_1^\partial$ and $y_2^\partial$ represent deviations from the periodic solutions $u_1^*$, $u_2^*$, $y_1^*$ and $y_2^*$ given periodic excitation $r_2^*$. With the change of variables $\boldsymbol{G}$ is unchanged but the nonlinearity $\boldsymbol{\phi}^\partial$ is itself periodic.
} 
\label{fig:Lurye_delta}
\end{figure}


Altshuller \cite{Altshuller:11,Altshuller:13} claims that it is ``well-known'' that such a periodic solution exists without even the requirement for a suitable multiplier. While there is a considerable literature on the existence of periodic solutions of differential and delay-differential equations (e.g.\cite{Halanay66,Yoshizawa66,Burton85}) we cannot find evidence for Altshuller's claim. R\u{a}svan \cite{Rasvan11} considers it an assumption. Yakubovich \cite{Yakubovich64} states without proof a similar result where the dynamics of $\boldsymbol{G}$ have a finite-dimensional state-space description and there exists a suitable Popov multiplier. R\u{a}svan \cite{Rasvan11} gives a generalisation to delay-differential systems with an incomplete proof.  In \cite{Jonsson03} existence is established using incremental IQCs; this approach rules out the use of dynamic multipliers for slope-restricted nonlinearities.




Nevertheless 
when the dynamics of $\boldsymbol{G}$ can be represented as a finite-dimensional state-space system, and when there exists a suitable OZF multiplier, the assumption can be verified from material in \cite{Yoshizawa66} with further restrictions on both the nonlinearity and the exogenous signal. 
One of the restrictions in the development in \cite{Yoshizawa66} is that the nonlinearity be Lipschitz continuous. 
We only consider the slope-restricted case which ensures Lipschitz continuity.
\begin{theorem}\label{thm:ss}
Consider a Lurye system where $\boldsymbol{\phi}\in\Phi^{sr}_k\cap\Phi^{ti}$ for some $k>0$. 
Let the dynamics of $\boldsymbol{G}$ be described by the finite-dimensional minimal state-space equations
\begin{align}
    \dot{x}(t) & =Ax(t)+Bu_1(t),\nonumber\\
    y_1(t) & = Cx(t)
\end{align}
where $A$ is Hurwitz.
Let $r_1=0$ and let $r_2$ be periodic with period $T>0$, non-zero and Lipschitz continuous. If there is 
an $\boldsymbol{M}\in\mathcal{M}$ suitable for~$1/k+\boldsymbol{G}$ then there exists a (not-necessarily unique or attracting) non-zero periodic solution with period $T$.  
\end{theorem}




\begin{proof}
Consider first the unforced case where $r_2=0$. Then the existence of a suitable OZF multiplier ensures the system is asymptotically stable in the large (\cite{vidyasagar93}, Theorem 46, Section 6.3). Furthermore, since its dynamics are time-invariant, it is uniform asymptotically stable in the large (\cite{Yoshizawa66}, Theorem 11.3). Hence the system satisfies all the conditions of Theorem 19.5 in \cite{Yoshizawa66} and there exists a Liapunov function $V(t,x(t))$ with the following properties.
\begin{enumerate}
\item There exist continuous, increasing, positive definite functions $a$ and $b$ such that
\begin{equation}
    a(\|x(t)\|)\leq V(t,x(t)) \leq b(\|x(t)\|),
\end{equation}
where $a(r)\rightarrow\infty$ as $r\rightarrow\infty$.  
\item There exists a continuous, positive definite function $c$ such that
\begin{align}
\varlimsup_{h\rightarrow 0^+}\frac{1}{h} 
\left \{ V(t+h,\right . & x(t)+h[Ax(t)-BN(t,Cx(t))])\nonumber\\
& \left . -V(t,x(t))\right \} \leq -c(\|x\|).
\end{align}
\end{enumerate}

Now we consider the same system but where $r_2$ is periodic with period $T>0$, non-zero, and Lipschitz continuous. There exists an $R$ such that for all $\|x(t)\|\geq R$ the function $V(t,x(t))$ maintains the two properties. 
Hence by Theorem 10.4 in \cite{Yoshizawa66} all solutions are uniform-ultimately bounded. It follows from the definition that they are also equi-ultimately bounded. 
Finally by Theorem 29.3 in \cite{Yoshizawa66} this implies there exists a periodic solution with period~$T$.
\end{proof}


We can generalise Theorem~\ref{thm:ss} to delay-differential equations as follows.
\begin{theorem}\label{thm:dd}
Consider a Lurye system where $\boldsymbol{\phi}\in\Phi^{sr}_k\cap\Phi^{ti}$ for some $k>0$. 
Let the dynamics of $\boldsymbol{G}$ be described by the stable delay-differential state-space equations
\begin{align}
    \dot{x}(t) & = \int_{t-h}^t A(t-\theta)x(\theta)\,d\theta + \int_{t-h}^t B(t-\theta)u_1(\theta)\,d\theta,\nonumber\\
    y_1(t) & = Cx(t).
\end{align}
Let $r_1=0$ and let $r_2$ be periodic with period $T>0$, non-zero and Lipschitz continuous. If there is 
an $\boldsymbol{M}\in\mathcal{M}$ suitable for~$1/k+\boldsymbol{G}$ then there exists a (not-necessarily unique or attracting) non-zero periodic solution with period $T$.  
\end{theorem}

\begin{proof}
The proof is similar in structure to that of Theorem~\ref{thm:ss}. 
Theorem 37.1 in \cite{Yoshizawa66} is the counterpart of Theorem 29.3 with the requirement that solutions are both uniform-bounded and uniform-ultimately bounded. Similarly Theorem 35.1 in \cite{Yoshizawa66} is the counterpart of Theorem 10.4 in \cite{Yoshizawa66} with a Liapunov-functional taking the place of Liapunov function and giving conditions that ensure solutions are both uniform-bounded and uniform-ultimately bounded. However there is no stated counterpart in \cite{Yoshizawa66} to Theorem 19.5 in \cite{Yoshizawa66}. 
Such a counterpart is given by Theorem 3.4 in \cite{Karafyllis11}; see also Theorem 4.4 in \cite{Haidar25}.
Similarly Sections 4.3(a) and 4.2 in \cite{Karafyllis11} provide the appropriate respective 
counterparts to Theorem 46, Section 6.3 in \cite{vidyasagar93} and Theorem 11.3 in \cite{Yoshizawa66}.
\end{proof}

In addition, we may state the following corollary to Theorem~\ref{thm:perexc}, irrespective of the state space structure.
\begin{corollary}\label{cor:lin}
Consider a Lurye system where $\boldsymbol{\phi}\in\Phi^{ti}$,  where $r_1=0$ and where $r_2$ is periodic (and non-zero) with period $T>0$. Suppose a linear solution is feasible. If there is 
an $\boldsymbol{M}\in\mathcal{A}_T$ suitable for~$\boldsymbol{G}$ then this solution is a global attractor. 

A similar result follows if, in addition, $\phi\in\Phi^{sr}_k$ for some $k>0$  
and there is an $\boldsymbol{M}\in\mathcal{A}_T$ suitable for~$1/k+\boldsymbol{G}$.
\end{corollary}
\begin{proof}
The linear solution must be periodic. Hence the result follows from Theorem~\ref{thm:perexc}.
\end{proof}

\begin{remark}
In Theorem~\ref{thm:ss} it is possible to relax the condition on $\boldsymbol{\phi}$ so that it is only necessarily slope-restricted for large values. The relaxation requires the state-space description to be minimal. Yakubovich \cite{Yakubovich64} states without proof a version of Theorem~\ref{thm:ss} with such a relaxation where there is a suitable Popov multiplier. R\u{a}svan \cite{Rasvan11} states with an incomplete proof a generalisation of the result in \cite{Yakubovich64} to delay-differential systems. 
\end{remark}

\subsection{Phase properties}

It is immediate from Definition~\ref{def:alt} that if $\boldsymbol{M}\in\mathcal{A}_T$ with impulse response (\ref{alt_imp}) then its frequency response
\begin{equation}\label{alt_freq}
    M(j\omega)=1-\sum_{n=-\infty}^{\infty}h_n e^{-j\omega nT}.
\end{equation}
 is periodic with period $2\pi/T$.
Hence we can say
\begin{theorem}
Let $\boldsymbol{G}$ be an LTI system with frequency response $G(j\omega)$. Suppose, given some $\omega>0$ and $T>0$,
\begin{equation}
\left |
    \angle G(j\omega) - \angle G(j\omega +2n\pi/T)
    \right | > \pi\text{ for some integer } n.
\end{equation}
Then there is no $\boldsymbol{M}\in\mathcal{A}_T$ suitable for $\boldsymbol{G}$.
\end{theorem}
\begin{proof}
Suppose $\boldsymbol{M}\in\mathcal{A}_T$ has frequency response (\ref{alt_freq}) and
\begin{equation}|\angle M(j\omega)G(j\omega)|\leq\pi/2.\end{equation}
Then
\begin{align}
    \angle M(j\omega +2n\pi/T) & G(j\omega +2n\pi/T)\nonumber \\
   & = \angle M(j\omega)G(j\omega +2n\pi/T).
\end{align}
So if $G(j\omega +2n\pi/T)-G(j\omega)>\pi$ then 
\begin{align}
    \angle M(j\omega +2n\pi/T)  G(j\omega +2n\pi/T) & > \angle M(j\omega)  G(j\omega)+\pi
    \nonumber\\
    & > \pi/2,
\end{align}
and if $G(j\omega +2n\pi/T)-G(j\omega)<-\pi$ then
\begin{align}
    \angle M(j\omega +2n\pi/T)  G(j\omega +2n\pi/T) & < \angle M(j\omega)  G(j\omega)-\pi
    \nonumber\\
    & < -\pi/2.
\end{align}
\end{proof}

Furthermore, it is immediate from Definitions~\ref{def:alt} and~\ref{def2a_d} that the Altshuller multipliers and the discrete-time OZF multipliers share the same frequency response. Hence the Altshuller multipliers inherit the phase limitations of the discrete-time OZF multipliers \cite{Zhang22}.
Hence we can say
\begin{theorem}\label{thm:plm}
Let $\boldsymbol{G}$ be an LTI system with frequency response $G(j\omega)$. Let $\beta>1$ and $\lambda_1,\ldots,\lambda_{\beta-1}\geq 0$ with at least one non-zero. If, for some $T>0$,
\begin{equation}
\sum_{r=1}^{\beta-1}\text{Re}\left \{\lambda_r G(j\omega_r)\right\} \leq \min_{l\in\mathbb{Z}}
\left [
\sum_{r=1}^{\beta-1}\text{Re}\left \{\lambda_r G(j\omega_r)e^{-j\omega_r l}\right \}
\right ]
\end{equation}
where 
$\omega_r = (-1)^{p_r}\frac{r\pi}{T\beta}+\frac{2n_r\pi}{T}$ 
for $r=1,\ldots\beta-1$, where either $p_r=0$ with $n_r\in\mathbb{Z}^+\cup\{0\}$ or $p_r=1$ with $n_r\in\mathbb{Z}^+$, then there is no Altshuller multiplier $\boldsymbol{M}\in\mathcal{A}_T$ with frequency response $M(j\omega)$ such that 
\begin{equation}
    \text{Re}\left \{M(j\omega)G(j\omega)\right \}>0 \text{ for all }\omega\in\mathbb{R}.
\end{equation}
\end{theorem}
\begin{proof}
Immediate from Theorem 2 in \cite{Zhang22}.
\end{proof}

If $n_r$ and $p_r$ are fixed then, as in \cite{Zhang22}, it is possible to test the conditions of Theorem~\ref{thm:plm} as the feasibility of a linear program. Although possibly conservative, useful results can also be obtained at single frequencies.
\begin{theorem}\label{thm:single}
If $\boldsymbol{M}\in\mathcal{A}_T$, for some $T>0$, then its frequency response satisfies, for any $a,b\in\mathbb{Z}^+$ with $b>1$,
\begin{align}
\left |\angle M \left (\frac{a}{b}\frac{2\pi j}{T}\right ) \right |  & \leq  \frac{\pi}{2}\left (1-\frac{2}{b}\right ).
\end{align}
Equivalently let $\boldsymbol{G}$ be an LTI system with frequency response $G(j\omega)$.  Suppose there exists some Altshuller multiplier $\boldsymbol{M}\in\mathcal{A}_T$ with frequency response $M(j\omega)$ such that 
\begin{equation}
    \text{Re}\left \{M(j\omega)G(j\omega)\right \}>0 \text{ for all }\omega\in\mathbb{R}.
\end{equation}
Then 
\begin{align}
\left |\angle G \left (\frac{a}{b}\frac{2\pi j}{T}\right ) \right |\leq   \pi\left (1-\frac{1}{b}\right ).
\end{align}
\end{theorem}
\begin{proof}
    Immediate from Theorem~3 in \cite{Zhang22}.
\end{proof}

\begin{corollary}\label{cor:incr}
Let $\boldsymbol{G}$ be an LTI system with frequency response $G(j\omega)$.  Suppose for all $T>0$ there exists an Altshuller multiplier $\boldsymbol{M}_T\in\mathcal{A}_T$ with frequency response $M_T(j\omega)$ such that 
\begin{equation}
    \text{Re}\left \{M_T(j\omega)G(j\omega)\right \}>0 \text{ for all }\omega\in\mathbb{R}.
\end{equation}
Then we must have 
\begin{equation}
\left |\angle G \left (j\omega\right )\right | \leq \frac{\pi}{2}
\end{equation}
at all frequencies $\omega$.
\end{corollary}
\begin{proof}
Setting $b=2$ in Theorem~\ref{thm:single} gives the requirement
\begin{equation}
\left|\angle M \left (\frac{n\pi j}{T}\right ) \right | =0 \text{ and }
\left |\angle G \left (\frac{n\pi j}{T}\right ) \right |\leq   \frac{\pi}{2}.
\end{equation}
for all $n\in\mathbb{Z}^+$ and for all $T>0$.
\end{proof}
We may interpret Corollary~\ref{cor:incr} as saying if we insist on a unique periodic solution at all frequencies, then we can do no better using multipliers than the circle criterion. This is consistent with Brockett's statement in \cite{Brockett:65} that ``the circle criterion has proven to be just the right tool to establish results on the
existence of unique steady-state oscillations.''

\section{The example of Fromion and Safonov \cite{Fromion04}}



Fromion and Safonov \cite{Fromion04} consider a Lurye system where the LTI system $\boldsymbol{G}$ has transfer function
\begin{equation}
G(s) = \frac{g}{(s^2+0.1+s)(s+100)}\text{ with } g=909.
\end{equation}
Their nonlinearity $\boldsymbol{\phi}\in\Phi^{sr}_1\cap\Phi^{ti}$ is smooth.
They give a suitable OZF multiplier for $1+\boldsymbol{G}$, specifically $\boldsymbol{M}\in\mathcal{M}$ with transfer function
\begin{equation}
M(s)=\frac{1+9s}{1+10^{-6}s}.
\end{equation}
The phases of $1+G(j\omega)$, $M(j\omega)$ and $M(j\omega)(1+G(j\omega))$ are all shown in Fig~\ref{fig01}.
\begin{figure}[tb]
\begin{center}
\includegraphics[width = \figwidth \columnwidth]{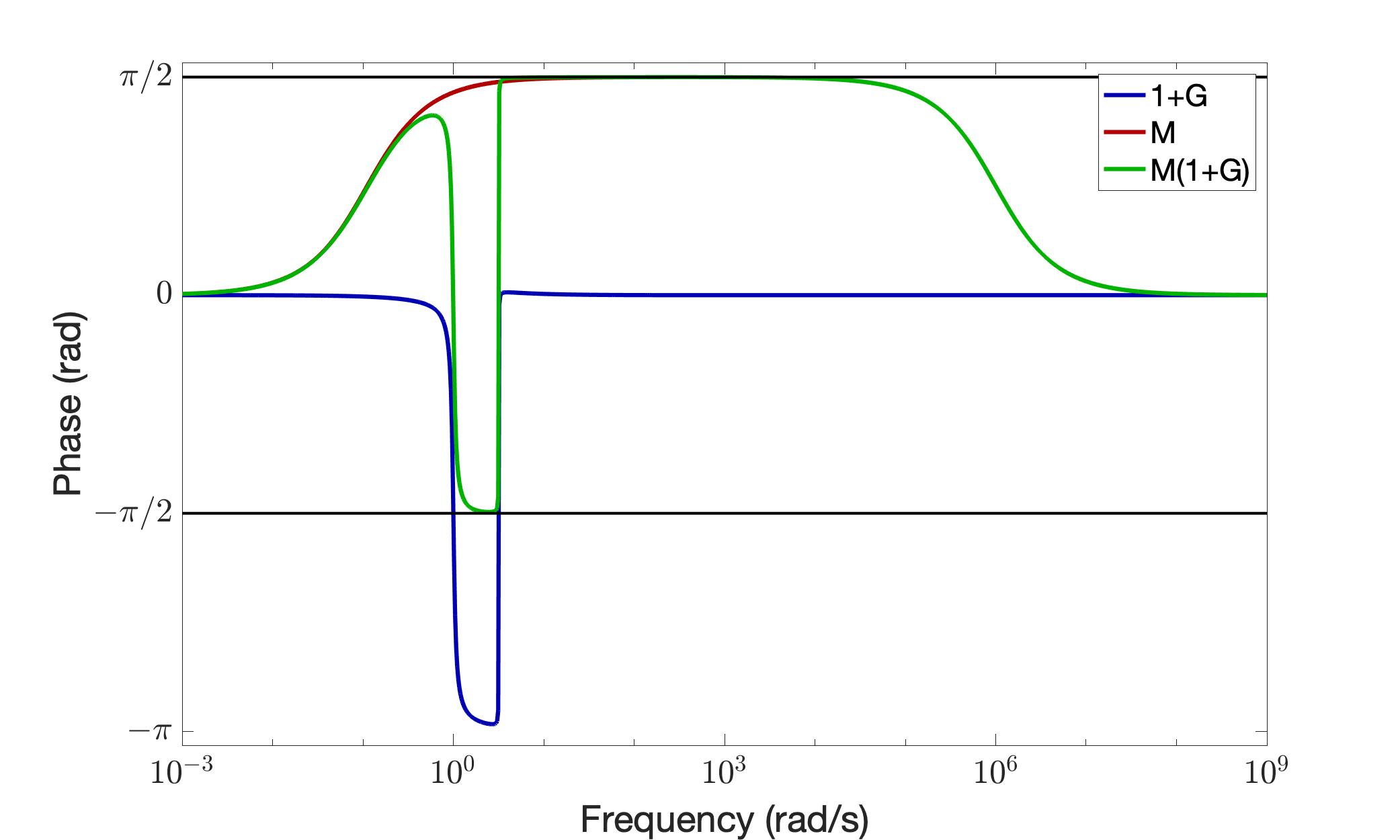}
\caption{Example from \cite{Fromion04} with gain 909. There is a suitable OZF multiplier (phase of its frequency response shown in red). 
}\label{fig01}
\end{center}
\end{figure}

Nevertheless when the exogenous signal is $r_2(t)=\sin(2t)$ (for~$t$ sufficiently large) then they show there are two periodic solutions, both locally attractive, and the system may reach either depending on initial conditions. 
It is straightforward to reproduce the simulation responses they report but with a standard saturation function characterised by
\begin{equation}\label{saturation}
Q(u(t)) = \left \{\begin{array}{ll}-1 & \text{when }u(t)\leq -1,\\
u(t) & \text{when } -1<u(t) < 1,\\
1 & \text{when }u(t)\geq 1.
\end{array}\right .
\end{equation}
Two responses, depending on initial conditions, are shown in Fig~\ref{fig:from}. It is noteworthy that one is the linear response while the other meets the saturation. 
\begin{figure}[ht]
\begin{center}
\includegraphics[width = \figwidth\columnwidth]{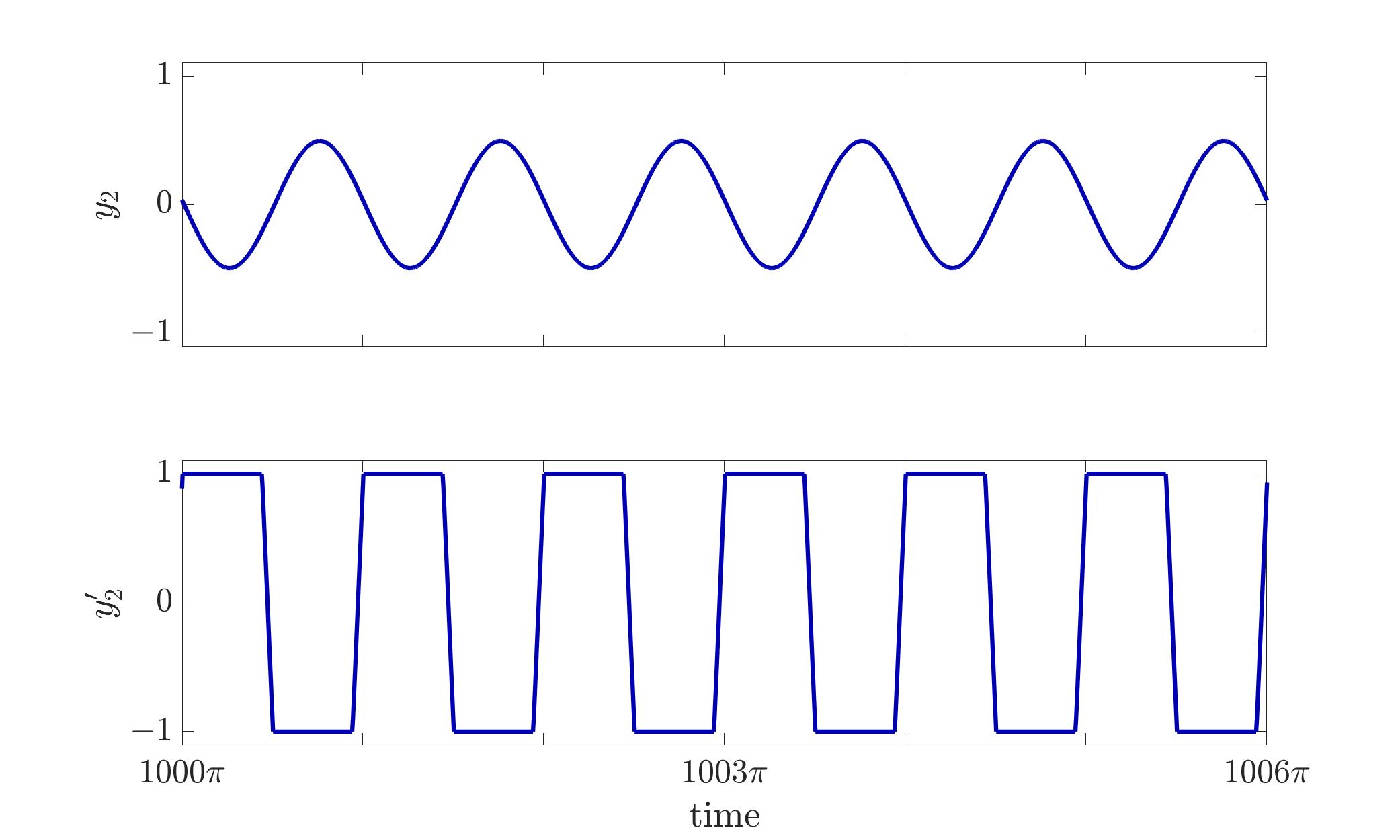}
\caption{Time responses from the example in \cite{Fromion04} with gain 909 and with a saturation nonlinearity. Similar behaviour is reported in \cite{Fromion04} where the signal before the nonlinearity is shown; here the response after the nonlinearity is shown. The first response $y_2$ is the linear response whereas the second $y_2'$ meets the saturation. All signals have period $\pi$, the period of the excitation signal.} \label{fig:from}
\end{center}
\begin{center}
\includegraphics[width = \figwidth \columnwidth]{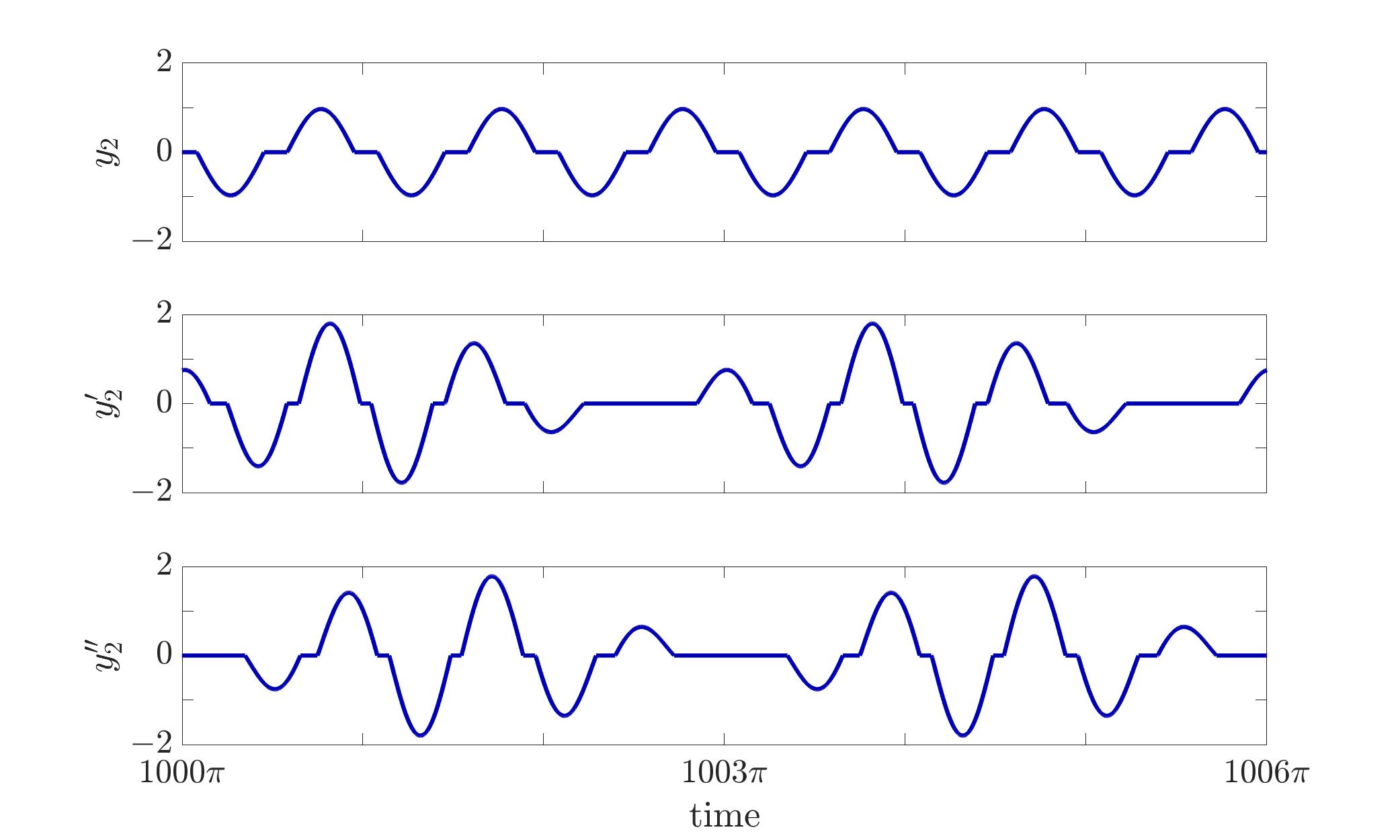}
\caption{Time responses from the example in \cite{Fromion04} with gain 909 but where the nonlinearity is a deadzone. The response may exhibit subharmonics in steady state. The excitation has period $\pi$ but the second and third responses have period $3\pi$}\label{fig:subh}
\end{center}
\end{figure}

Simulations also indicate that such a system may exhibit subharmonics. Figure~\ref{fig:subh} shows three responses (specifically $y_2$) when the nonlinearity is a deadzone characterised by
\begin{equation}\label{deadzone}
Q(u(t)) = \left \{\begin{array}{ll}u(t)+w & \text{when }u(t)\leq -w,\\
0 & \text{when } -w<u(t) < w,\\
u(t)-w & \text{when }u(t)\geq w,
\end{array}\right .
\end{equation}
with $w=0.5$.
The first response is periodic with period $\pi$, as is the excitation.  The second and third responses are periodic with period $3\pi$. If we label them $y_2'$ and $y_2''$ respectively, then  $y'_2(t+\pi)\neq y_2'(t)$ except at isolated points, and similarly for $y_2''$. The signals are related as $y_2''(t)=-y_2'(t-\pi)$.


If we set $0<g<20.77$ (working to two decimal places) then $\boldsymbol{G}$ satisfies the circle criterion 
and the Lurye system is guaranteed incremental stable. By contrast, if $g>73.37$ then the phase of $1+\boldsymbol{G}$ exceeds the phase limitation of Theorem~\ref{thm:single} for period $T=\pi$ at 1.2 rad/s (see Fig~\ref{fig03}). Hence there can be no suitable Altshuller multiplier $\boldsymbol{M}\in\mathcal{A}_\pi$ for $g>73.37$.


If we set $g=50$ then multipliers $\boldsymbol{M}_\theta\in\mathcal{A}_{\theta \pi}$ are all suitable for $1+\boldsymbol{G}$ with $1\leq \theta\leq 1.08$. Fig~\ref{fig05} illustrates this for $\theta = 1, 1.01, \ldots, 1.08$. This says there is a unique globally attractive solution when the exogenous signal is $r_2=\sin(2t/\theta)$ over the same range of $\theta$. NB these values were obtained by hand; it is possible to find suitable higher order Altshuller multipliers with larger values of $g$. We make no attempt to optimise values here. However the phase of $1+G(j\omega)$ drops below $-\pi/2$ when $\omega = 1.01$ rad/s. This says there can be no $\boldsymbol{M}\in\mathcal{A}_{0.99\pi}$.

\begin{figure}[ht]
\begin{center}
\includegraphics[width = \figwidth \columnwidth]{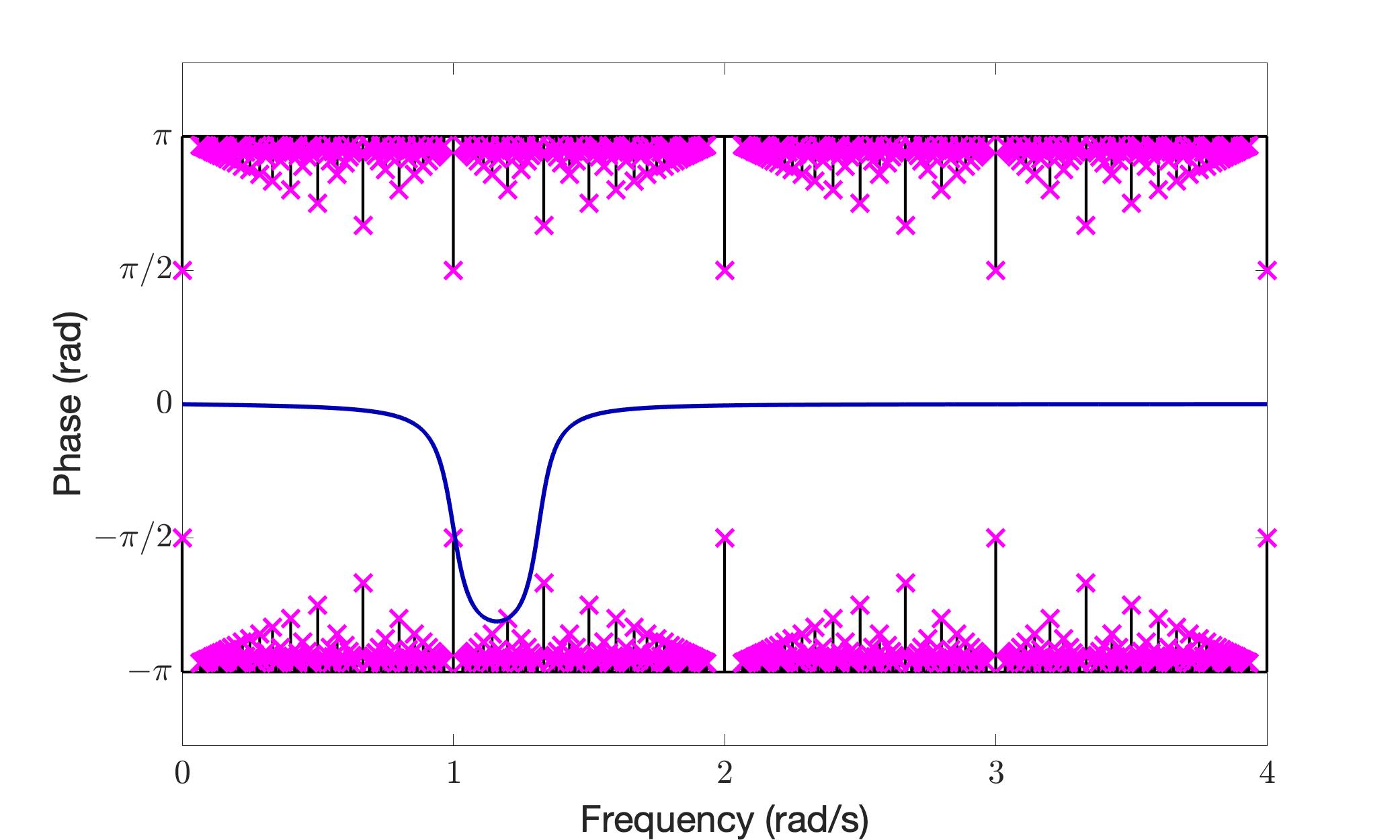}
\caption{Example from \cite{Fromion04} but with gain 73.37. The phase of $1+G$ touches the limitation at frequency $1.2$ rad/s. There can be no suitable Altshuller multiplier $\boldsymbol{M}\in\mathcal{A}_\pi$ when the gain is higher.}\label{fig03}
\end{center}
\begin{center}
\includegraphics[width = \figwidth \columnwidth]{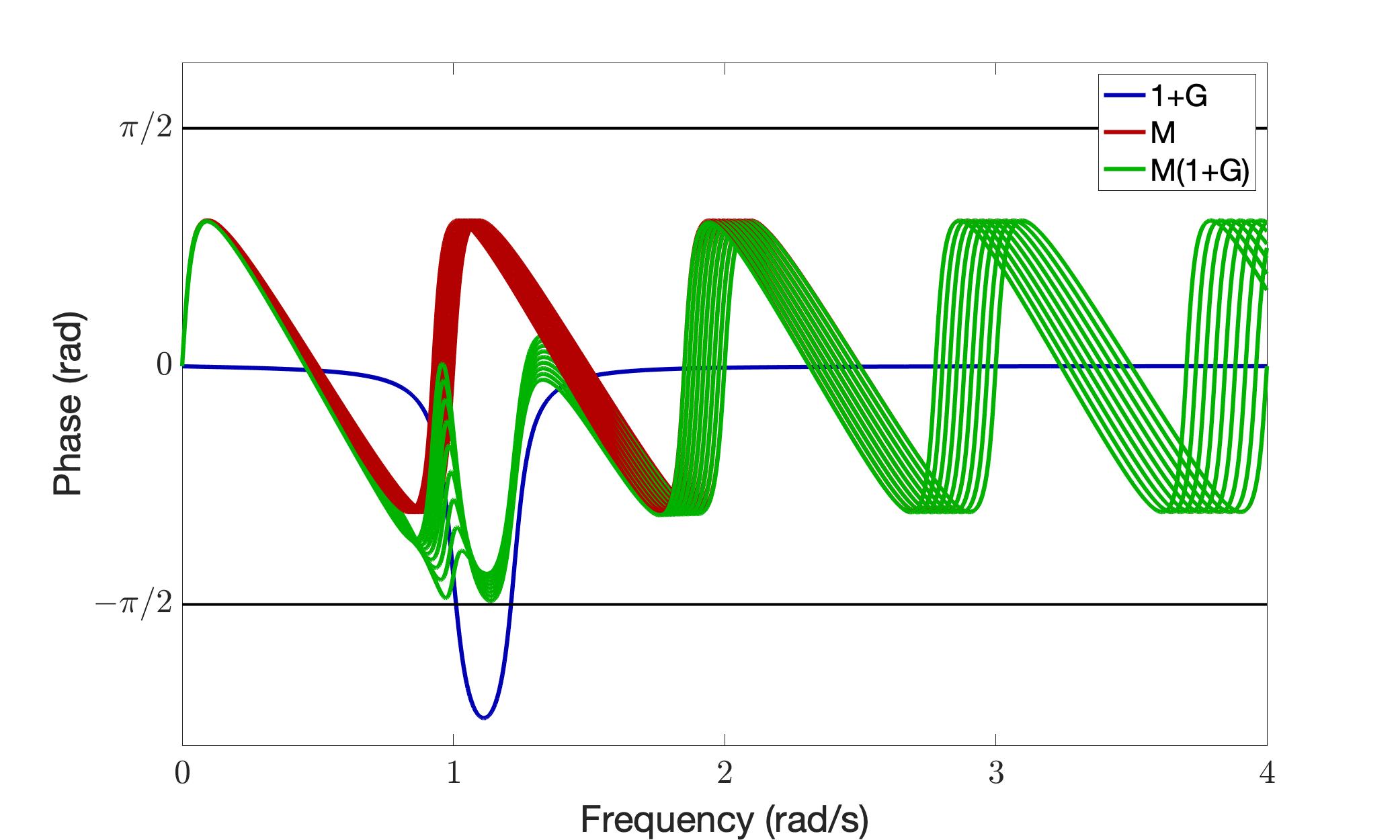}
\caption{Example from \cite{Fromion04} but with gain 50. The Altshuller multipliers $\boldsymbol{M}\in\mathcal{A}_{\theta \pi}$ with frequency responses $M(j\omega)=1-0.82e^{-2\theta j\omega\pi}$ with $1\leq \theta \leq 1.08$ are all suitable for $1+\boldsymbol{G}$. However the phase of $1+\boldsymbol{G}$ drops below $-\pi/2$ at $w\approx 1.01$ rad/s so there can be no suitable $\boldsymbol{M}\in\mathcal{A}_{0.99\pi}$}\label{fig05}
\end{center}
\end{figure}

\section{Discrete-time example}

Here we consider a Lurye system where the LTI plant has dynamics with discrete-time transfer function
\begin{equation}
G(z) = g \frac{2z+0.92}{z(z-0.5)}
\end{equation}
with $g$ taking the values $g=0.6$, $0.7$, $0.8$, $0.9$ or $1.0$. The nonlinearity $\boldsymbol{\phi}\in\Phi^{sr}_1\cap\Phi^{ti}$. 
We have considered this Lurye system previously with $g=1$ \cite{Heath15,Heath22} and in our preliminary results with $g=0.6$, $0.8$ and $1$ \cite{Heath:24}. 

\subsection{Multiplier analysis}

In the following, we will consider the existence of multipliers for each gain in turn. With one exception we will consider only single parameter multipliers, and give values correct to two decimal place. The upper bounds given on gains are not least upper bounds. Results are summarised in Table~\ref{tab:ex2}.

\begin{description}
\item{When $g=0.6$ the} system satisfies the circle criterion; specifically $\text{Re}\left [1+G(e^{j\omega})\right ]>0$ for all $\omega\in[0,2\pi)$. This establishes that the closed-loop system is continuous and both FGS and FGOS with gain bounded above by 
$h\leq 12.76$ (applying (\ref{first_gamma}) with the identity multiplier $\boldsymbol{M}=\boldsymbol{I}$ whose frequency response is $M(e^{j\omega})=1$).  Note that $\boldsymbol{I}\in\mathcal{M}^d$, $\boldsymbol{I}\in\mathcal{M}^d_{\text{odd}}$ and $\boldsymbol{I}\in\mathcal{A}^d_N$ for all $N\in\mathbb{Z}^+$. The upper bound on the gain can be reduced using dynamic multipliers. For example the  multiplier in $\mathcal{M}^d$ with transfer function $M(z)=1-0.68z^{-1}$ allows us to bound the offset gain by $3.76$. The multiplier in $\mathcal{M}^d_{odd}$ with transfer function $M(z)=1+0.57z$ allows us to bound the gain (but not the offset gain) by $3.39$. 
\item{When $g=0.7$} the system no longer satisfies the circle criterion. Nevertheless the multiplier in $\mathcal{M}^d$ with transfer function $M(z)=1-0.91z^{-1}$ establishes that the system is FGOS with gain bounded above by $h\leq 5.73$. 
The multiplier in $\mathcal{M}^d_{odd}$ with transfer function $M(z)=1+0.64z$ allows us to bound the gain (but not the offset gain) by $4.69$. 
We can find suitable Altshuller multipliers in $\mathcal{A}^d_3$ and $\mathcal{A}^d_5$ with transfer functions $M(z)=1-0.2z^3$ and $M(z)=1-0.16z^{-5}-0.04z^{10}$ respectively 
(see Fig~\ref{fig:g7;a5}).
    \begin{figure}[htbp]
        \begin{center}
        \includegraphics[width=\figwidth\columnwidth]{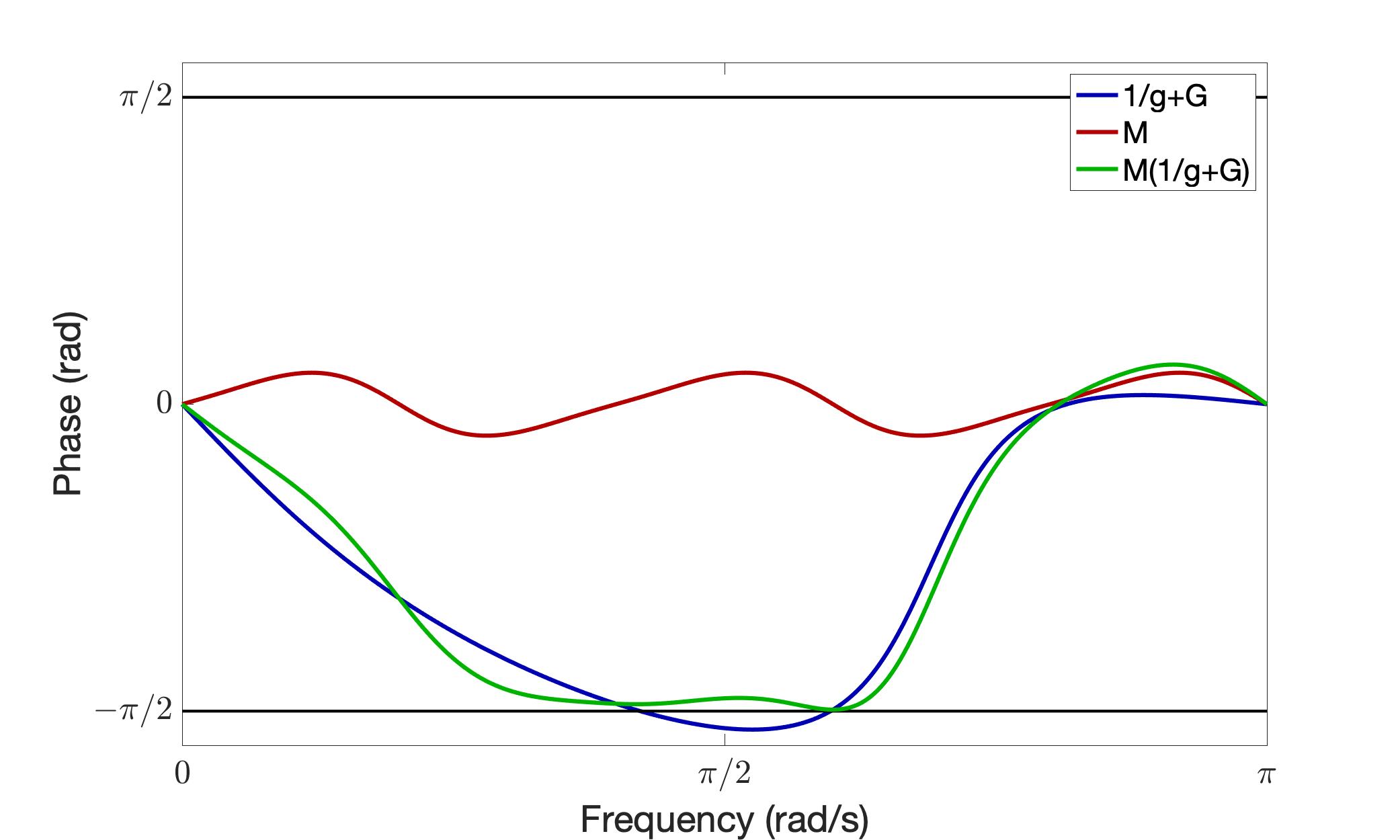}
        \end{center}
        \caption{For the discrete-time example with $g=0.7$ there is a suitable multiplier $\boldsymbol{M}\in\mathcal{A}_5^d$ with transfer function $M(z)=1-0.16z^{-5}-0.04z^{10}$.}\label{fig:g7;a5}
    \end{figure}
    But there is no suitable Altshuller multiplier in $\mathcal{A}^d_N$ when  $N$ is even since $\angle (1+G(e^{j\pi/2}))<-\pi/2$ (see Fig~\ref{fig:A67} for the case $N=6$). Similarly there is no suitable Altshuller multiplier in $\mathcal{A}^d_{N}$ when  $N$ is odd and $N\geq 7$ since $\angle (1+G(e^{j(n+1)\pi/(2n+1)}))<-\pi/2$ when $n\geq 3$. 
    \begin{figure}[t]
        \begin{center}
        \includegraphics[width=\figwidth \columnwidth]{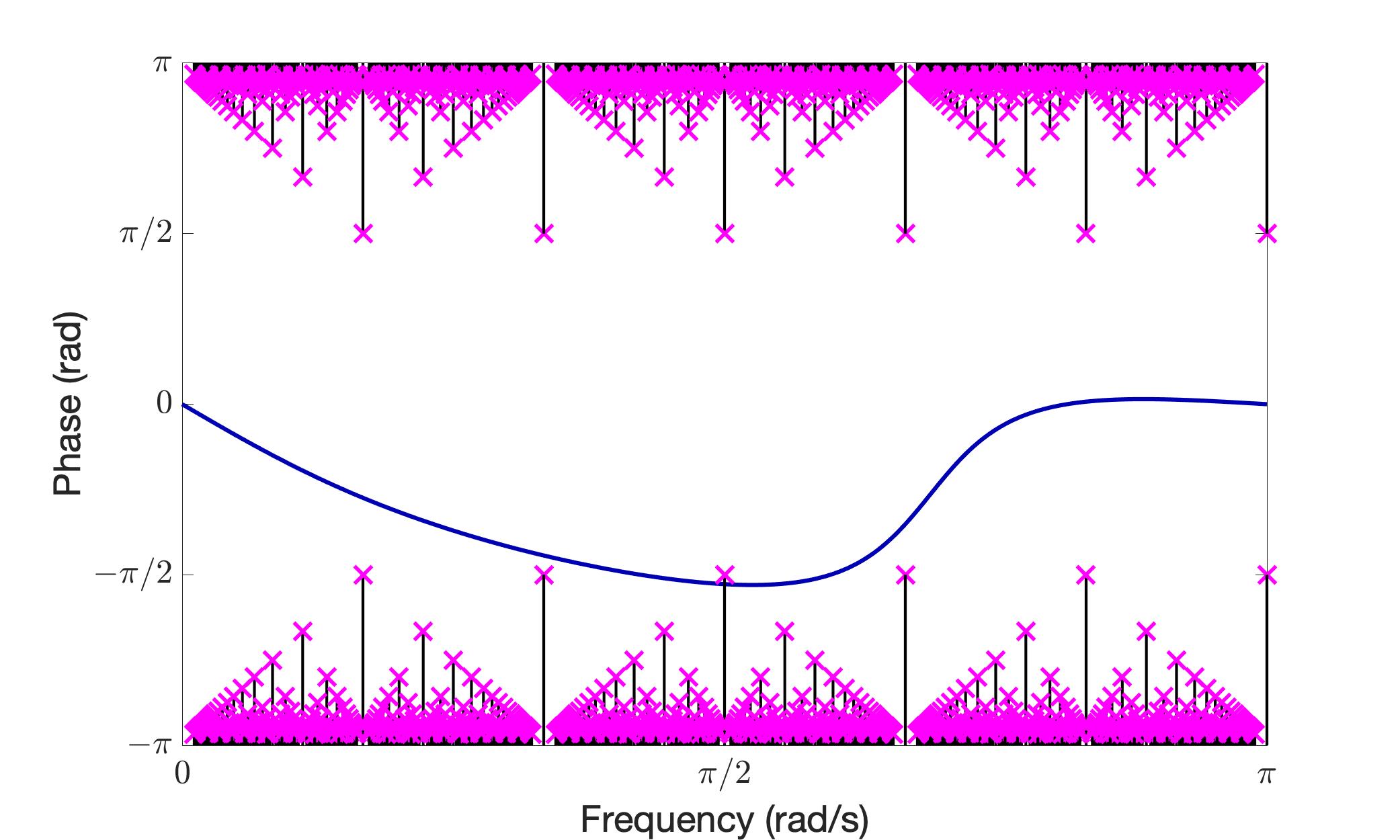}
        \end{center}
        \caption{For the discrete-time example with $g=0.7$ there is no suitable $\boldsymbol{M}\in\mathcal{A}^d_6$. The phase of $1+G(e^{j\omega})$ lies below $-\pi/2$ at $\omega=\pi/2$. The same single frequency phase limitation occurs for any $\mathcal{A}^d_N$ with $N$ even.}\label{fig:A67}
        \begin{center}
        \includegraphics[width=\figwidth \columnwidth]{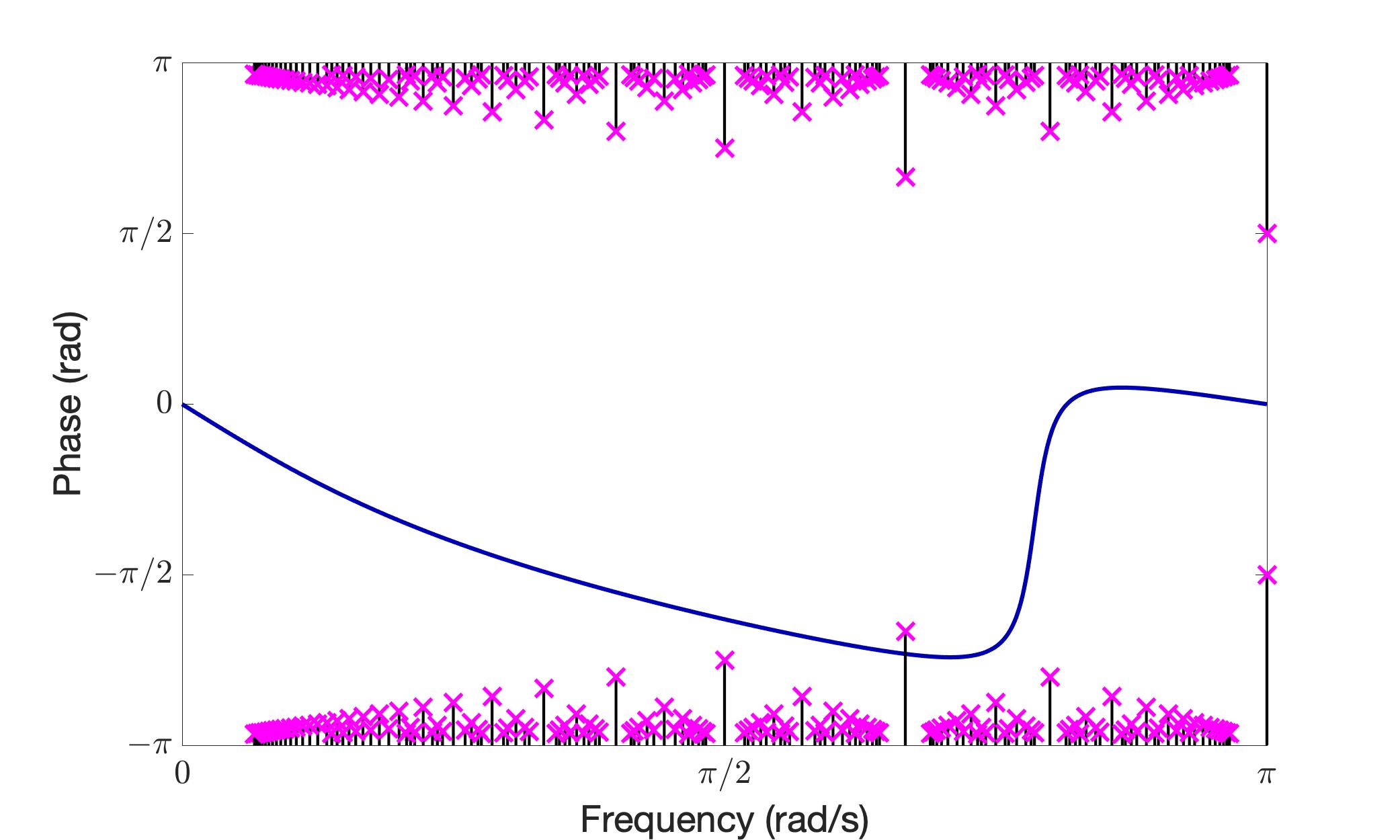}
        \end{center}
        \caption{For the discrete-time example with $g=1$ there is no suitable $\boldsymbol{M}\in\mathcal{M}^d$. The phase of $1+G(e^{j\omega})$ lies below $-2\pi/3$ at $\omega=2\pi/3$.}\label{fig:A110}
    \end{figure}
\item{When $g=0.8$} the multiplier in $\mathcal{M}^d$ with transfer function $M(z)=1-0.99z^{-1}$ establishes that the system is FGOS with gain  bounded above by $h\leq 10.96$. The multiplier in $\mathcal{M}^d_{odd}$ with transfer function $M(z)=1+0.72z$ allows us to bound the gain (but not the offset gain) by $7.07$. 
    But there is no suitable Altshuller multiplier in $\mathcal{A}^d_{N}$ for any  $N\geq 2$. 
\item{When $g=0.9$} the multiplier in $\mathcal{M}^d$ with transfer function $M(z)=1-0.99z^{-1}$ establishes that the system is FGOS with gain bounded above by $h\leq 121.28$. The multiplier in $\mathcal{M}^d_{odd}$ with transfer function $M(z)=1+0.79z$ allows us to bound the gain (but not the offset gain) by $12.42$. 
    But there is no suitable Altshuller multiplier in $\mathcal{A}^d_{N}$ for any  $N\geq 2$.
\item{When $g=1$} we find
     $\angle [1+G(e^{2\pi j/3}) ]  = -\pi+\text{atan} \frac{31\sqrt{3}}{48}$ \\
    $< -\frac{2\pi}{3}$. See Fig~\ref{fig:A110}.
It follows that there is no $\boldsymbol{M}\in\mathcal{M}^d$ suitable for $1+\boldsymbol{G}$. Nevertheless there are multipliers in $\mathcal{M}^d_{odd}$ suitable for $1+\boldsymbol{G}$. It follows that the Lurye system is FGS, but not necessarily FGOS in this case. For example the multiplier $M(z)=1+0.87z$  establishes an upper bound $h \leq 31.74$ on the $\ell_2$ gain.
\end{description}

\begin{table*}
    \centering
    \begin{tabular}{|c|c|c|c|c|c|}
    \hline
    $g$ & $0.6$    & $0.7$  & $0.8$ & $0.9$  & $1.0$ \\ \hline \hline
    Continuous   & Y  &  N & N  & N & N \\ \hline \hline
    FGOS     & Y & Y  & Y & Y & N\\
    \hline
    Bound  & $3.76$ & $5.73$ & $10.96$ & $121.28$ & \\ 
    Multiplier & $1-0.68z^{-1}$ & $1-0.91z^{-1}$ & $1-0.99z^{-1}$  & $1-0.99z^{-1}$ & \\ 
    \hline \hline
     FGS    & Y & Y  & Y & Y & Y\\
     \hline
     Bound   & $3.39$ & $4.69$ & $7.07$ & $12.42$ & $31.74$ \\ 
     Multiplier & $1+0.57z$   & $1+0.64z$   & $1+0.72z$     & $1+0.79z$  & $1+0.87z$  \\ \hline \hline
     Altshuller $\mathcal{A}^d_3$ & Y   & Y  &  N & N & N \\
     Mutliplier & $1$ & $1-0.2z^3$ & & &  \\ \hline
     Altshuller $\mathcal{A}^d_5$ & Y   & Y  &  N & N & N \\
     Mutliplier & $1$ & $1 -0.16z^{-5} - 0.04z^{10}$ & &  & \\ \hline
    \end{tabular}
    \caption{Suitable multipliers and gain bounds for the discrete-time example. All the bounds are on gains from $r_2$ to $y_2$.}
    \label{tab:ex2}
\end{table*}

\subsection{Response to periodic excitation}

In \cite{Heath:24} we considered the step response when the nonlinearity is a saturation and the gain $g$ is either $0.6$, $0.8$ or $1.0$. We found there was a qualitative difference in behaviour between the response when $g=0.8$, i.e. when the closed-loop system is guaranteed FGOS and when $g=1.0$, i.e. when it is only guaranteed FGS. When the closed-loop system is not FGOS, the power must be normalized around $0$ rather than around the steady state values. The response is indeed  ``critically sensitive to small changes in inputs'' \cite{Zames66a} when the gain is $g=1.0$.

Here we focus on the response when the nonlinearity is a deadzone and the excitation is periodic. Let the nonlinearity be a deadzone function (\ref{deadzone}) with $w=0.2$. Let the excitation be periodic with period 5. Specifically:
\begin{equation}
\left [
\begin{array}{c}
r_2(1+5(n-1))\\
r_2(2+5(n-1))\\
r_2(3+5(n-1))\\
r_2(4+5(n-1))\\
r_2(5+5(n-1))
\end{array}
\right ]
= \left [
\begin{array}{r}
1.0\\ 0.6\\-0.6\\ -1.0\\ 0.0
\end{array}
\right ] \text{ with }n\in\mathbb{Z}^+.
\end{equation}
When $g=0.7$ there is a suitable Altshuller multiplier in $\mathcal{A}^d_5$. In simulation the outputs $y_2$ converge to steady state values
\begin{equation}
\left [
\begin{array}{c}
\bar{y}_2(1)\\
\bar{y}_2(2)\\
\bar{y}_2(3)\\
\bar{y}_2(4)\\
\bar{y}_2(5)
\end{array}
\right ]
= \left [
\begin{array}{r}
0.2282\\
   -0.2861\\
   -0.6895\\
         0
         \\
    0.7464
\end{array}
\right ],
\end{equation}
accurate to four decimal places. We find $\|r_2\|_P =0.7376$, $\|\bar{y}_2\|_P = 0.4830$ and hence $\|\bar{y}_2\|_P/\|r_2\|_P < 5.73$
(the bound given in Table~\ref{tab:ex2}).

By contrast, when $g=0.9$ there is no suitable Altshuller multiplier in $\mathcal{A}^d_5$ and the output $y_2$ does not appear to settle into a stable cycle. Fig~\ref{Fig:chaos1} shows the absolute value of the discrete Fourier transform of $10^6$ data points after a simulation where the first $1,000$ data points are discarded. There is a significant subharmonic response with period $40$, shown in Fig~\ref{Fig:chaos2}. But the ``leakage'' terms are nontrivial and apparently persistent (Figs~\ref{Fig:chaos2} and~\ref{Fig:chaos3}). If we write $y_2=y_2^p+y_2^v$ where $y_2^p$ is the periodic component then we measure $\|y_2^p\|_P=0.41$ and $\|y_2^v\|_P=5.1\times 10^{-4}$ (both to two significant figures).

We can write the system in state-space form as
\begin{align}
x(n+1)= & \left [\begin{array}{cc}0.5 & 0\\1 & 0\end{array}\right ]x(n)\nonumber\\
& -g\left [\begin{array}{c}2\\ 0\end{array}\right ]Q\left (
\left [\begin{array}{cc}1 & 0.46\end{array}\right ]x(n)+r_2(n)\right ),\nonumber \\
y_1(n)  =&  \left [\begin{array}{cc}1 & 0.46\end{array}\right ]x(n).
\end{align}
A standard measure for the Liapunov exponent of the state's evolution \cite[pp.116-117]{Sprott03} yields $0.012$ (using natural logarithms and to three decimal places) consistently.  A positive value indicates chaotic dynamics \cite{Sprott03}. 

    \begin{figure}[tbp]
        \begin{center}
        \includegraphics[width=\figwidth\columnwidth]{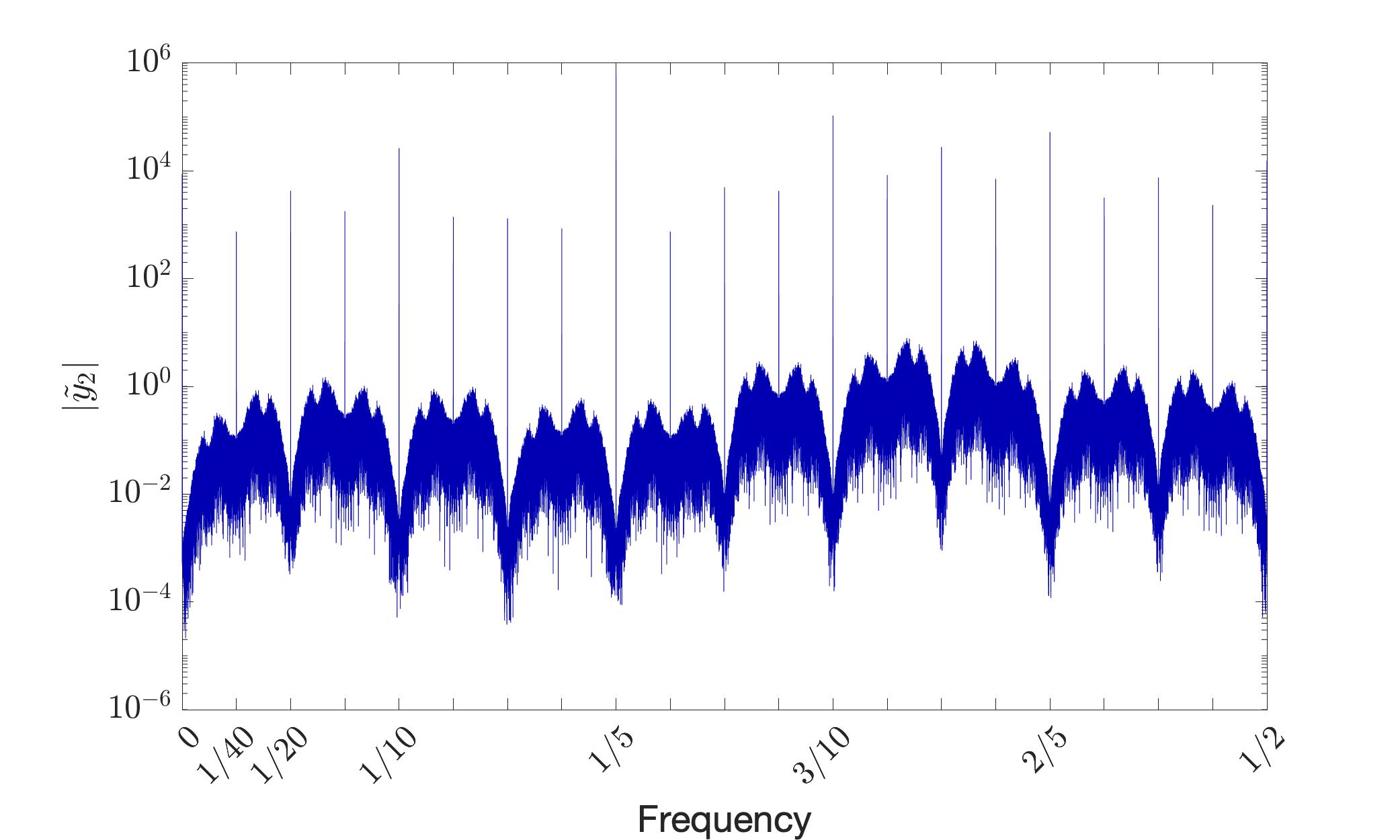}
        \end{center}
        \caption{Discrete time example with periodic excitation, period $5$, and deadzone nonlinearity. The plot shows the absolute value of the fast Fourier transform of $y_2$ (log scale) measured over $10^6$ data points. There is a significant response with period $40$ but also nontrivial variation at all frequencies.}\label{Fig:chaos1}
        \begin{center}
        \includegraphics[width=\figwidth\columnwidth]{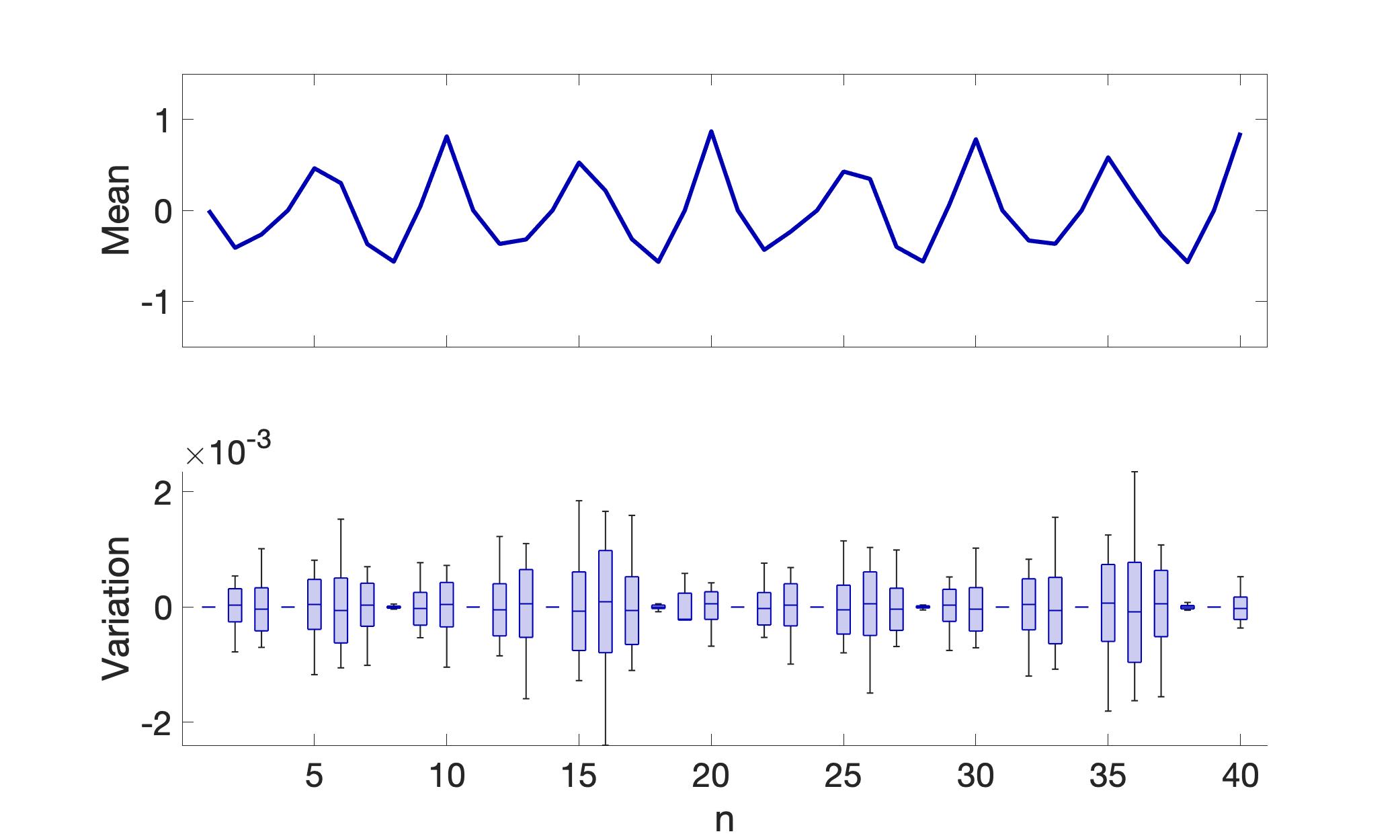}
        \end{center}
        \caption{The top plot shows the mean response over a period of $40$. The bottom box and whisker chart shows variation at each $n$ ranging from $1$ to $40$ - i.e. the variation of $y_2(n+40(k-1)$ with $k$ taking values from $1$ to $25,000$.}\label{Fig:chaos2}
        \begin{center}
        \includegraphics[width=\figwidth\columnwidth]{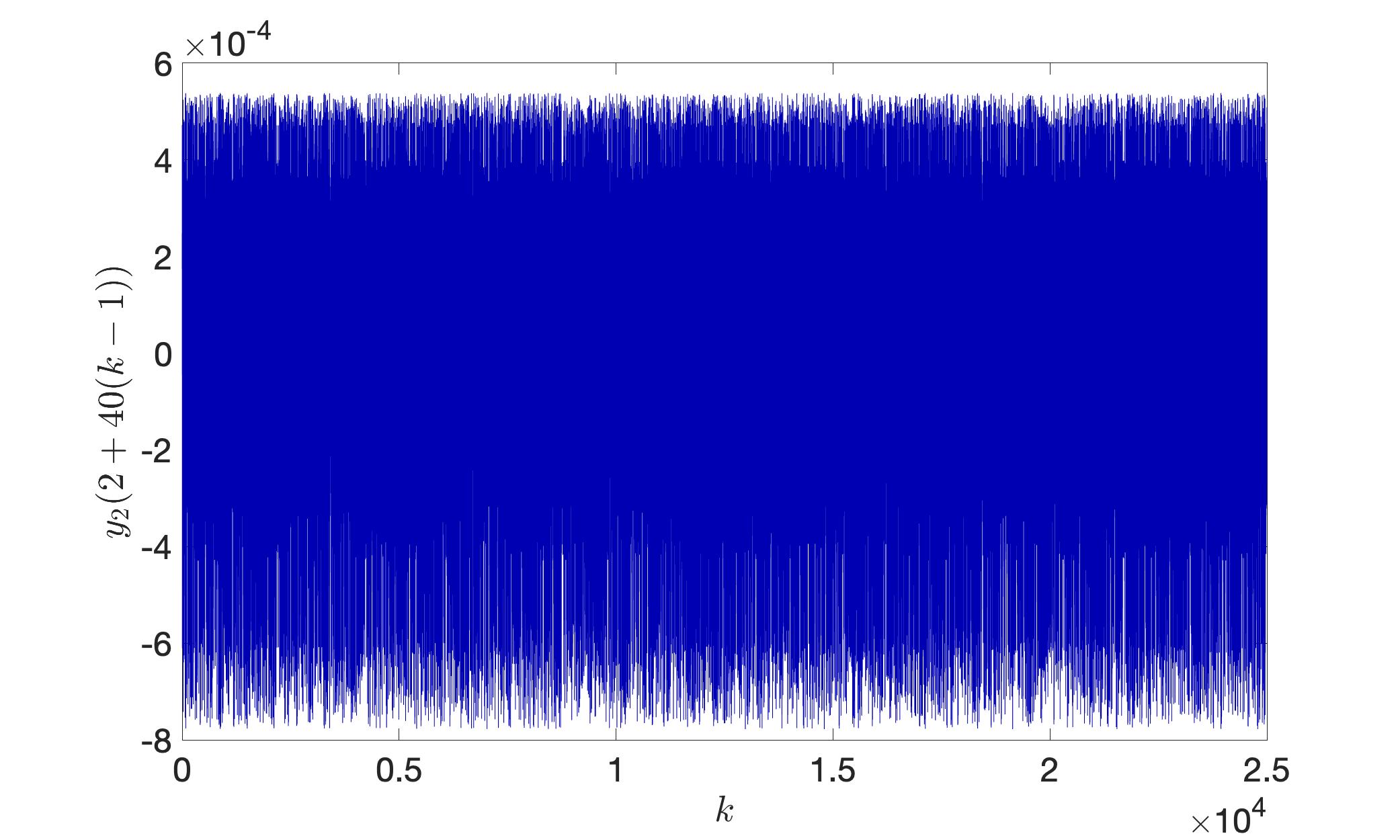}
        \end{center}
        \caption{Value of $y_2(2+40(k-1))$. The variation is apparently persistent.}\label{Fig:chaos3}
    \end{figure}

\section{Conclusion}
We have considered the behaviour of Lurye systems with time-invariant nonlinearities and exogenous signals that have finite power  but not finite energy. It is known \cite{Kulkarni:2002a,Fromion04} that dynamic multipliers do not guarantee continuity in closed loop. This might be considered an Achilles' heel of stability analysis based on dynamic multipliers: our examples suggest that period doubling or chaotic behaviour can occur with periodic excitation.  While  Lurye systems with chaotic dynamics have been widely reported in the literature (see, e.g., \cite{GENESIO1992,MIRANDAVILLATORO2018}), to the best of the authors’ knowledge, this is the first example where chaotic behaviour occurs despite $\ell_2$ (or $\mathcal{L}_2$)  input-output stability being guaranteed.

We have shown that the bounds on the $\mathcal{L}_2$ gains provided by suitable OZF multipliers also bound the power gain. Furthermore, the existence of a suitable OZF multiplier in $\mathcal{M}$ (but not $\mathcal{M}_\text{odd}-\mathcal{M}$) ensures a Lurye system has a unique steady state map and is FGOS. This in turn ensures that if the exogenous signal has small power measured around some bias then the output also has small power measured around a uniquely determined bias.

Nevertheless if the excitation is periodic then the discontinuities may be significant. We have shown examples where, even without noise, the closed-loop system may exhibit subharmonic or chaotic responses. 
For such cases the Altshuller multipliers \cite{Altshuller:11,Altshuller:13} can be used to guarantee better behaviour for excitation with specified periods. We have shown that these 
 multipliers and their properties can be derived using classical methods. This allows us to generalise their application to Lurye systems, both in terms of the LTI element $\boldsymbol{G}$ and the nonlinear element $\boldsymbol{\phi}$. We have established conditions where the assumptions about existence in \cite{Altshuller:11,Altshuller:13} can be justified.

Theorems~\ref{thm:ss} and ~\ref{thm:dd} and Corollary~\ref{cor:lin} lead us to the conjecture that a more general result may be true:
\begin{conjecture}\label{conj}
    Consider a Lurye system where $\boldsymbol{\phi}\in\Phi^{ti}$, where $r_1=0$  and where $r_2$ is periodic (and non-zero) with period $T>0$. If there is 
    an Altshuller multiplier $\boldsymbol{M}\in\mathcal{A}_T$ suitable for~$\boldsymbol{G}$ (or for $1/k+\boldsymbol{G}$ when $\boldsymbol{\phi}\in\Phi^{sr}_k$ for some $k>0$) then there exists a (not-necessarily unique or attracting) non-zero periodic solution with period $T$.  
\end{conjecture}
\begin{remark}
 Since $\mathcal{A}_T\subset\mathcal{M}$ then if $\boldsymbol{M}\in\mathcal{A}_T$ satisfies the conditions of Theorem~\ref{thm:perexc} then it suffices for Theorems~\ref{thm:ss},~\ref{thm:dd} and Conjecture~\ref{conj}.
 \end{remark}

 Following the structure of the proof of Theorem~\ref{thm:ss}, we make the following two Conjectures which together are sufficient for Conjecture~\ref{conj} to be true. 
 
 \begin{subtheorem}{conjecture}
 \begin{conjecture}\label{conj2a}
    Consider a Lurye system where $\boldsymbol{\phi}\in\Phi^{ti}$, where $r_1=0$  and where $r_2\in\mathcal{L}_2\cap\mathcal{L}_{\infty}$. If there is 
    an OZF multiplier $\boldsymbol{M}\in\mathcal{M}$ suitable for~$\boldsymbol{G}$ (or for $1/k+\boldsymbol{G}$ when $\boldsymbol{\phi}\in\Phi^{sr}_k$ for some $k>0$)  then $y_2\in\mathcal{L}_2\cap\mathcal{L}_{\infty}$.
 \end{conjecture}
 \begin{conjecture}\label{conj2b}
    Consider a Lurye system where $\boldsymbol{\phi}\in\Phi^{ti}$, where $r_1=0$  and where $r_2$ is periodic (and non-zero) with period $T$. If all
     solutions are uniform-bounded and uniform ultimate-bounded then there exists a (not-necessarily unique or attracting) non-zero periodic solution with period $T$.  
 \end{conjecture}
  \end{subtheorem}
  Both Conjectures~\ref{conj2a} and~\ref{conj2b} are of interest in their own right. Conjecture~\ref{conj2a} seems physically intuitive: it would be disconcerting were it not true. Conjecture~\ref{conj2b} begs the question whether the fixed point theory exploited by \cite{Yoshizawa66}  can be used without internal structure.

The close relation between the Altshuller multipliers and the discrete-time OZF multipliers means the Altshuller multipliers inherit the phase limitations of \cite{Zhang22}. These phase limitations in turn shed light on the relation with known results about dynamic multipliers, continuity of the input-output map and incremental stability \cite{Brockett66a,Kulkarni:2002a,Fromion04}. Specifically, while the Altshuller multipliers can be used to ensure a unique periodic solution for excitation at a specific frequency or range of frequencies, they cannot be used to ensure a unique periodic solution for excitation at all frequencies. 

We have indicated that it is straightforward to extend the results both to discrete-time systems and to multivariable systems. 
It remains open to develop efficient algorithms both to search for Altshuller multipliers and to test for the phase limitation of Theorem~\ref{thm:plm}.

\bibliographystyle{IEEEtranS}
\bibliography{ok}



\begin{IEEEbiography}[{\includegraphics[width=1in,height=1.25in,clip,keepaspectratio]{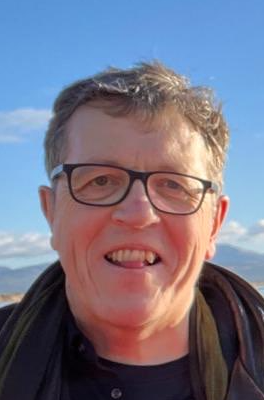}}]{William P. Heath}
 %
 received an M.A. in mathematics from the University of Cambridge, U.K. and both an M.Sc. and Ph.D. in systems and control from U.M.I.S.T., U.K.  He is head of the School of Computer Science and Engineering at Bangor University and chair of the UK Automatic Control Council. He has previously held positions at the University of Manchester, at the University of Newcastle (Australia) and  at Lucas Automotive. 

\end{IEEEbiography}

\begin{IEEEbiography}[{\includegraphics[width=1in,height=1.25in,clip,keepaspectratio]{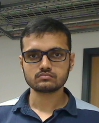}}]{Sayar Das} is a PhD student in Control Systems at The University of Manchester. He received his B.E. in Electrical Engineering from Jadavpur University, in 2019 and his M.S(R) from the Department of Electrical Engineering, Indian Institute of Technology Delhi (IIT Delhi), specializing in Control Systems and Automation, in 2023. His current research interests include absolute stability, passivity and multiplier theory.
\end{IEEEbiography}

\begin{IEEEbiography}[{\includegraphics[width=1in,height=1.25in,clip,keepaspectratio]{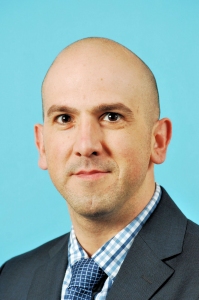}}]{Joaquin Carrasco}
                is a Reader at the Control Systems Centre, Department of Electrical and Electronic Engineering, University of Manchester, UK. He was born in Abarán, Spain, in 1978. He received the B.Sc. degree in physics and the Ph.D. degree in control engineering from the University of Murcia, Murcia, Spain, in 2004 and 2009, respectively. From 2009 to 2010, he was with the Institute of Measurement and Automatic Control, Leibniz Universität Hannover, Hannover, Germany. From 2010 to 2011, he was a research associate at the Control Systems Centre, School of Electrical and Electronic Engineering, University of Manchester, UK. His current research interests include absolute stability, multiplier theory, and robotics applications.
\end{IEEEbiography}

\end{document}

%% file: lure_v4.tex
\ifx\JPicScale\undefined\def\JPicScale{1}\fi
\unitlength \JPicScale mm
\begin{picture}(60,25)(0,0)
\linethickness{0.3mm}
\put(0,20){\line(1,0){7.5}}
\put(7.5,20){\vector(1,0){0.12}}
\linethickness{0.3mm}
\put(10,5){\line(1,0){15}}
\linethickness{0.3mm}
\put(10,5){\line(0,1){12.5}}
\put(10,17.5){\vector(0,1){0.12}}
\linethickness{0.3mm}
\put(12.5,20){\line(1,0){12.5}}
\put(25,20){\vector(1,0){0.12}}
\linethickness{0.3mm}
\put(35,20){\line(1,0){15}}
\linethickness{0.3mm}
\put(50,7.5){\line(0,1){12.5}}
\linethickness{0.3mm}
\put(35,5){\line(1,0){12.5}}
\put(35,5){\vector(-1,0){0.12}}
\put(0,25){\makebox(0,0)[cc]{$r_1$}}

\put(42.5,25){\makebox(0,0)[cc]{$y_1$}}

\linethickness{0.3mm}
\put(25,25){\line(1,0){10}}
\put(25,15){\line(0,1){10}}
\put(35,15){\line(0,1){10}}
\put(25,15){\line(1,0){10}}
\linethickness{0.3mm}
\put(25,0){\line(1,0){10}}
\put(25,0){\line(0,1){10}}
\put(35,0){\line(0,1){10}}
\put(25,10){\line(1,0){10}}
\put(30,5){\makebox(0,0)[cc]{$\boldsymbol\phi$}}

\put(30,20){\makebox(0,0)[cc]{$\boldsymbol{G}$}}

\linethickness{0.3mm}
\put(10,20){\circle{5}}

\put(12.5,15){\makebox(0,0)[cc]{$-$}}

\put(17.5,25){\makebox(0,0)[cc]{$u_1$}}

\linethickness{0.3mm}
\put(50,5){\circle{5}}
\put(50,7.5){\vector(0,-1){0.12}}

\linethickness{0.3mm}
\put(52.5,5){\line(1,0){7.5}}
\put(52.5,5){\vector(-1,0){0.12}}
\put(60,2.5){\makebox(0,0)[cc]{$r_2$}}

\put(42.5,2.5){\makebox(0,0)[cc]{$u_2$}}

\put(17.5,2.5){\makebox(0,0)[cc]{$y_2$}}

\end{picture}

%% file: Lure_delta1.tex
\ifx\JPicScale\undefined\def\JPicScale{1}\fi
\unitlength \JPicScale mm
\begin{picture}(60,40)(0,0)
\linethickness{0.3mm}
\linethickness{0.3mm}
\put(10,20){\line(1,0){15}}
\linethickness{0.3mm}
\put(10,20){\line(0,1){15}}
\linethickness{0.3mm}
\put(10,35){\line(1,0){15}}
\put(25,35){\vector(1,0){0.12}}
\linethickness{0.3mm}
\put(35,35){\line(1,0){15}}
\linethickness{0.3mm}
\put(50,22.5){\line(0,1){12.5}}
\linethickness{0.3mm}
\put(35,20){\line(1,0){12.5}}
\put(35,20){\vector(-1,0){0.12}}
\put(50,22.5){\vector(0,-1){0.12}}

\put(42.5,40){\makebox(0,0)[cc]{$y_1$}}

\linethickness{0.3mm}
\put(25,40){\line(1,0){10}}
\put(25,30){\line(0,1){10}}
\put(35,30){\line(0,1){10}}
\put(25,30){\line(1,0){10}}
\linethickness{0.3mm}
\put(25,25){\line(1,0){10}}
\put(25,15){\line(0,1){10}}
\put(35,15){\line(0,1){10}}
\put(25,15){\line(1,0){10}}
\put(30,20){\makebox(0,0)[cc]{$\boldsymbol\phi$}}

\put(30,35){\makebox(0,0)[cc]{$-\boldsymbol{G}$}}
\put(30,5){\makebox(0,0)[cc]{$\Delta$}}

\linethickness{0.3mm}


\put(17.5,40){\makebox(0,0)[cc]{$-u_1$}}

\linethickness{0.3mm}
\put(50,20){\circle{5}}

\linethickness{0.3mm}
\put(52.5,20){\line(1,0){7.5}}
\put(52.5,20){\vector(-1,0){0.12}}
\put(57.5,17.5){\makebox(0,0)[cc]{$r_2$}}

\put(42.5,17.5){\makebox(0,0)[cc]{$u_2$}}

\put(17.5,17.5){\makebox(0,0)[cc]{$y_2$}}

\linethickness{0.3mm}
\put(25,10){\line(1,0){10}}
\put(25,0){\line(0,1){10}}
\put(35,0){\line(0,1){10}}
\put(25,0){\line(1,0){10}}
\put(35,5){\line(1,0){25}}
\put(10,5){\line(1,0){15}}
\put(25,5){\vector(1,0){0.12}}
\put(10,5){\line(0,1){20}}
\put(60,5){\line(0,1){15}}

\end{picture}

%% file: lurye_delta.tex
\ifx\JPicScale\undefined\def\JPicScale{1}\fi
\unitlength \JPicScale mm
\begin{picture}(60,25)(0,0)
\linethickness{0.3mm}
\linethickness{0.3mm}
\put(10,5){\line(1,0){15}}
\linethickness{0.3mm}
\put(10,5){\line(0,1){12.5}}
\put(10,17.5){\vector(0,1){0.12}}
\linethickness{0.3mm}
\put(12.5,20){\line(1,0){12.5}}
\put(25,20){\vector(1,0){0.12}}
\linethickness{0.3mm}
\put(35,20){\line(1,0){15}}
\linethickness{0.3mm}
\put(50,7.5){\line(0,1){12.5}}
\linethickness{0.3mm}
\put(35,5){\line(1,0){12.5}}
\put(35,5){\vector(-1,0){0.12}}

\put(47.5,25){\makebox(0,0)[cc]{$y_1^\partial=y_1-y_1^*$}}

\linethickness{0.3mm}
\put(25,25){\line(1,0){10}}
\put(25,15){\line(0,1){10}}
\put(35,15){\line(0,1){10}}
\put(25,15){\line(1,0){10}}
\linethickness{0.3mm}
\put(25,0){\line(1,0){10}}
\put(25,0){\line(0,1){10}}
\put(35,0){\line(0,1){10}}
\put(25,10){\line(1,0){10}}
\put(30,5){\makebox(0,0)[cc]{$\boldsymbol\phi^\partial$}}

\put(30,20){\makebox(0,0)[cc]{$\boldsymbol{G}$}}

\linethickness{0.3mm}
\put(10,20){\circle{5}}

\put(12.5,15){\makebox(0,0)[cc]{$-$}}

\put(12.5,25){\makebox(0,0)[cc]{$u_1^\delta=u_1-u_1^*$}}

\linethickness{0.3mm}
\put(50,5){\circle{5}}
\put(50,7.5){\vector(0,-1){0.12}}

\linethickness{0.3mm}
\put(52.5,5){\line(1,0){7.5}}
\put(52.5,5){\vector(-1,0){0.12}}
\put(60,7.5){\makebox(0,0)[cc]{$r_r-r_2^*$}}

\put(47.5,0){\makebox(0,0)[cc]{$u_2^\partial = u_2-u_2^*$}}

\put(12.5,2.5){\makebox(0,0)[cc]{$y_2^\partial= y_2-y_2^*$}}

\end{picture}

%% file: ok.bib
@ARTICLE{Angeli02,
  author={Angeli, D.},
  journal={IEEE Transactions on Automatic Control}, 
  title={A {L}yapunov approach to incremental stability properties}, 
  year={2002},
  volume={47},
  number={3},
  pages={410-421},
  doi={10.1109/9.989067}}

@article{Carrasco12,
title = {Factorization of multipliers in passivity and {IQC} analysis},
journal = {Automatica},
volume = {48},
number = {5},
pages = {909-916},
year = {2012},
author = {Joaquín Carrasco and William P. Heath and Alexander Lanzon}
}

@article{Kharitenko23,
title = {Time-varying {Z}ames–{F}alb multipliers for {LTI} Systems are superfluous},
journal = {Automatica},
volume = {147},
pages = {110577},
year = {2023},
author = {Andrey Kharitenko and Carsten Scherer}
}

@article{Bertolin22,
author = {Bertolin, Ariádne L. J. and Oliveira, Ricardo C. L. F. and Valmorbida, Giorgio and Peres, Pedro L. D.},
title = {Control design of uncertain discrete-time {L}ur'e systems with sector and slope bounded nonlinearities},
journal = {International Journal of Robust and Nonlinear Control},
volume = {32},
number = {12},
pages = {7001-7015},
year = {2022}
}

@ARTICLE{Chaffey23,
  author={Chaffey, Thomas and Forni, Fulvio and Sepulchre, Rodolphe},
  journal={IEEE Transactions on Automatic Control}, 
  title={Graphical Nonlinear System Analysis}, 
  year={2023},
  volume={68},
  number={10},
  pages={6067-6081},
  doi={10.1109/TAC.2023.3234016}}

@ARTICLE{Brockett66a,
  author={Brockett, R.},
  journal={IEEE Transactions on Automatic Control}, 
  title={The status of stability theory for deterministic systems}, 
  year={1966},
  volume={11},
  number={3},
  pages={596-606},
  doi={10.1109/TAC.1966.1098354}}

@BOOK{desoer75,
  title = {Feedback systems: input-output properties},
  publisher = {Academic Press, reprinted SIAM 2009},
  year = {1975},
  author = {Desoer, Charles A. and Vidyasagar, M.},
  isbn = {0122120507}
}

@inproceedings{Fromion04,
title = {Popov-{Z}ames-{F}alb multipliers and continuity of the input/output map},
year = {2004},
booktitle = {6th IFAC Symposium on Nonlinear Control Systems (NOLCOS), Stuttgart, Germany},
author = {Vincent Fromion and Michael G. Safonov},
}

@ARTICLE{Fromion96,
  author={Fromion, V. and Monaco, S. and Normand-Cyrot, D.},
  journal={IEEE Transactions on Automatic Control}, 
  title={Asymptotic properties of incrementally stable systems}, 
  year={1996},
  volume={41},
  number={5},
  pages={721-723},
  doi={10.1109/9.489210}}

@ARTICLE{Megretski97,
  author={Megretski, A. and Rantzer, A.},
  journal={IEEE Transactions on Automatic Control}, 
  title={System analysis via integral quadratic constraints}, 
  year={1997},
  volume={42},
  number={6},
  pages={819-830},
}

@article{Heath15,
title = {Second-order counterexamples to the discrete-time {K}alman conjecture},
journal = {Automatica},
volume = {60},
pages = {140-144},
year = {2015},
issn = {0005-1098},
author = {William Paul Heath and Joaquin Carrasco and Manuel {de la Sen}},
}

@ARTICLE{Heath22,
  author={Heath, William Paul and Carrasco, Joaquin and Altshuller, Dmitry A.},
  journal={IEEE Transactions on Automatic Control}, 
  title={Multipliers for Nonlinearities With Monotone Bounds}, 
  year={2022},
  volume={67},
  number={2},
  pages={910-917},
}

@INPROCEEDINGS{Heath:24,
  author={Heath, William P. and Carrasco, Joaquin},
  booktitle={2024 IEEE 63rd Conference on Decision and Control (CDC)}, 
  title={Multiplier analysis of {L}urye systems with power signals}, 
  year={2024},
  volume={},
  number={},
  pages={5864-5869},
  doi={10.1109/CDC56724.2024.10885824}}

@article{Megretski04,
  title={A guide to {IQC} $\beta$: A {M}atlab toolbox for robust stability and performance analysis},
  author={Megretski, A and Kao, C and Jonsson, U and Rantzer, A},
  journal={Technical Report, MIT},
  year={2004}
}

@inproceedings{Kao04,
  title={A {M}ATLAB toolbox for robustness analysis},
  author={Kao, Chung-Yao and Megretski, Alexandre and Jonsson, UT and Rantzer, Anders},
  booktitle={IEEE International Conference on Robotics and Automation},
  year={2004}
}

@ARTICLE{Mari96,
  author={Mari, J.},
  journal={IEEE Transactions on Automatic Control}, 
  title={A counterexample in power signals space}, 
  year={1996},
  volume={41},
  number={1},
  pages={115-116},
}

@inproceedings{Sepulchre22,
title = {On the incremental form of dissipativity},
booktitle = {25th International Symposium on Mathematical Theory of Networks and Systems MTNS},
author = {Rodolphe Sepulchre and Thomas Chaffey and Fulvio Forni},
year = {2022}
}

@article{Veenman16,
title = {Robust stability and performance analysis based on integral quadratic constraints},
journal = {European Journal of Control},
volume = {31},
pages = {1-32},
year = {2016},
issn = {0947-3580},
doi = {https://doi.org/10.1016/j.ejcon.2016.04.004},
author = {Joost Veenman and Carsten W. Scherer and Hakan Köroğlu},
}

@BOOK{vidyasagar93,
 author = {M. Vidyasagar},
 title = {Nonlinear systems analysis, 2nd edition},
 publisher = {Prentice-Hall International Editions, (reprinted SIAM 2002)},
 year = {1993},
}

@inproceedings{Waitman17,
title = {Incremental stability of {L}ur’e systems through piecewise-affine approximations},
year = {2017},
booktitle = {20th IFAC World Congress},
author = {Sérgio Waitman and Laurent Bako and Paolo Massioni and Gérard Scorletti and Vincent Fromion},
}

@article{Zames68,
  title={Stability conditions for systems with monotone and slope-restricted nonlinearities},
  author={Zames, George and Falb, PL},
  journal={SIAM Journal on Control},
  volume={6},
  number={1},
  pages={89--108},
  year={1968},
  publisher={SIAM}
}

@BOOK{Zhou96,
 author = {K. Zhou and J. C. Doyle and K. Glover},
 title = {Robust and optimal control},
 publisher = {Prentice-Hall},
 year = {1996},
}

@article{Willems68,
  title={Some new rearrangement inequalities having application in stability analysis},
  author={Willems, Jan C and Brockett, Roger W},
  journal={IEEE Transactions on Automatic Control},
  volume={13},
  number={5},
  pages={539--549},
  year={1968},
}

@book{Willems71,
    author = {J. C. Willems},
    title = {The analysis of feedback systems},
    publisher = {MIT Press},
    year = {1971},
}

@ARTICLE{Zhang22,
  author={Zhang, Jingfan and Carrasco, Joaquin and Heath, William Paul},
  journal={IEEE Transactions on Automatic Control}, 
  title={Duality Bounds for Discrete-Time {Z}ames–{F}alb Multipliers}, 
  year={2022},
  volume={67},
  number={7},
  pages={3521-3528},
}

@article{Turner12,
author = {Turner, Matthew C. and Kerr, Murray L.},
title = {Gain bounds for systems with sector bounded and slope-restricted nonlinearities},
journal = {International Journal of Robust and Nonlinear Control},
volume = {22},
number = {13},
pages = {1505-1521},
year = {2012}
}

@ARTICLE{OShea67,
  author={O'Shea, R.},
  journal={IEEE Transactions on Automatic Control}, 
  title={An improved frequency time domain stability criterion for autonomous continuous systems}, 
  year={1967},
  volume={12},
  number={6},
  pages={725-731}}

@article{Veenman14,
author = {Veenman, Joost and Scherer, Carsten W.},
title = {{IQC}-synthesis with general dynamic multipliers},
journal = {International Journal of Robust and Nonlinear Control},
volume = {24},
number = {17},
pages = {3027-3056},
year = {2014}
}

@book{Partington04,
    author = {J. R. Partington},
    title = {Linear operators and linear systems: an analytical approach to control theory},
    publisher = {CUP},
    year = {2004},
}

@article{Wang:TAC,
author = {Joaquin Carrasco and William Paul Heath and Jingfan Zhang and Nur Syazreen Ahmad and 
Shuai Wang},
title = {Convex searches for discrete-time {Z}ames-{F}alb multipliers},
journal = {IEEE Transactions on Automatic Control},
  year={2020},
  volume={65},
  number={11},
  pages={4538-4553},
}

@article{Safonov2000,
author = {Safonov, Michael G. and Kulkarni, Vishwesh V.},
title = {Zames–{F}alb multipliers for {MIMO} nonlinearities},
journal = {International Journal of Robust and Nonlinear Control},
volume = {10},
number = {11-12},
pages = {1025-1038},
year = {2000}
}

@ARTICLE{Kulkarni2002,
  author={Kulkarni, V.V. and Safonov, M.G.},
  journal={IEEE Transactions on Automatic Control}, 
  title={All multipliers for repeated monotone nonlinearities}, 
  year={2002},
  volume={47},
  number={7},
  pages={1209-1212},
}

@BOOK{Altshuller:13,
author = {D. Altshuller},
title = {Frequency Domain Criteria for Absolute Stability: A Delay-integral-quadratic Constraints Approach},
publisher = {Springer},
year = {2013},
}

@article{Altshuller:11,
author = {D. A. Altshuller},
title = {Delay-integral-quadratic constraints and stability multipliers for systems with {MIMO} nonlinearities},
year = {2011},
journal = {IEEE Transactions on Automatic Control},
volume = {56},
number = {4},
pages = {738--747},
}

@ARTICLE{Carrasco:13,
author={J. Carrasco and W. P. Heath and A. Lanzon},
title={Equivalence between classes of multipliers for slope-restricted nonlinearities},
journal={Automatica},
year={2013},
volume={49},
number={6},
pages={1732-1740},
}

@ARTICLE{Brockett:65,
	author = {Brockett, R.W. and Willems, J.L.},
	title = {Frequency Domain Stability Criteria-Part I},
	year = {1965},
	journal = {IEEE Transactions on Automatic Control},
	volume = {10},
	number = {3},
	pages = {255 – 261}}

@ARTICLE{Kulkarni:2002a,
  author={Kulkarni, V.V. and Safonov, M.G.},
  journal={IEEE Transactions on Automatic Control}, 
  title={Incremental positivity nonpreservation by stability multipliers}, 
  year={2002},
  volume={47},
  number={1},
  pages={173-177},
  doi={10.1109/9.981740}}

@article{Yakubovich64,
    author = {Yakubovich, V. A.},
    journal = {Automation and Remote Control},
    volume = {25},
    number = {7},
    pages = {905-916},
    title = {The matrix-inequality method in the theory of the stability of nonlinear
control systems: 1. The absolute stability of forced vibrations},
    year = {1964}
}

@INPROCEEDINGS{Rasvan11,
  author={Răsvan, Vladimir},
  booktitle={15th International Conference on System Theory, Control and Computing}, 
  title={Forced periodic oscillations and almost linear behavior}, 
  year={2011},
  volume={},
  number={},
  pages={1-6},
  doi={}}

@BOOK{Yoshizawa66,
    author = {T. Yoshizawa},
    title = {Stability theory by {L}iapunov’s second method},
    publisher = {Mathematical Society of Japan},
    year = {1966}
    }

@BOOK{Burton85,
    author = {T. A. Burton},
    title = {Stability and periodic solutions of ordinary and functional differential equations},
    publisher = {Academic Press},
    year = {1985}
    }

@article{Haidar25,
    author = {Haidar, I. and  Mason, P.},
    title = {Converse {L}yapunov results for uniform stability properties.},
    journal = {J. Optim. Theory Appl.},
    volume = {206},
    number = {39},
    year = {2025},
    doi = {https://doi.org/10.1007/s10957-025-02698-1}
}

@book{Karafyllis11,
    author = {Karafyllis, I. and Jiang, Z.P.},
    title = {Stability and stabilization of nonlinear systems},
    publisher = {Springer Science \& Business Media},
    year = 2011
    }

@book{Halanay66,
    author = {A. Halanay},
    title = {Differential equations: stability, oscillations, time lags},
    publisher = {Academic Press},
    year = {1966}
}

@article{Su23,
title = {On the necessity and sufficiency of discrete-time {O}’{S}hea–{Z}ames–{F}alb multipliers},
journal = {Automatica},
volume = {150},
pages = {110872},
year = {2023},
issn = {0005-1098},
author = {Lanlan Su and Peter Seiler and Joaquin Carrasco and Sei Zhen Khong},
}

@ARTICLE{Zames66a,
  author={Zames, G.},
  journal={IEEE Transactions on Automatic Control}, 
  title={On the input-output stability of time-varying nonlinear feedback systems. {P}art one: conditions derived using concepts of loop gain, conicity, and positivity}, 
  year={1966},
  volume={11},
  number={2},
  pages={228-238},
  doi={10.1109/TAC.1966.1098316}}

@article{Willems72,
    author = {J. C. Willems},
    title = {Dissipative dynamical systems. {P}art {I}: General theory},
    journal = {Arch. Rational Mech. Anal.},
    volume = {45},
    pages = {321-351},
    year = {1972}
}

@article{Bertolin25,
author = {A. L. J. Bertolin and G. Valmorbida and R. C. L. F. Oliveira and P. L. D. Peres},
title = {Output-feedback controllers with guaranteed $\mathcal{L}$2-gain for continuous-time {L}ur'e systems using noncausal {Z}ames–{F}alb multipliers},
journal = {International Journal of Control},
volume = {98},
number = {5},
pages = {1177--1190},
year = {2025},
publisher = {Taylor \& Francis}
}

@book{Sprott03,
    author = {J. C. Sprott},
    title = {Chaos and Time Series Analysis},
    publisher = {Oxford University Press},
    year = {2003}
}

@article{MIRANDAVILLATORO2018,
title = {Analysis of Lur’e dominant systems in the frequency domain},
journal = {Automatica},
volume = {98},
pages = {76-85},
year = {2018},
issn = {0005-1098},
doi = {https://doi.org/10.1016/j.automatica.2018.09.007},
author = {Félix A. Miranda-Villatoro and Fulvio Forni and Rodolphe J. Sepulchre},
}

@ARTICLE{Drummond2025,
  author={Drummond, Ross and Guiver, Chris and Turner, Matthew C.},
  journal={IEEE Transactions on Automatic Control}, 
  title={A note on incremental stability of externally positive Lurie systems}, 
  year={2025},
  volume={},
  number={},
  pages={1-8},
  doi={10.1109/TAC.2025.3615246}}

@ARTICLE{Su2025,
  author={Su, Lanlan and Khong, Sei Zhen},
  journal={IEEE Transactions on Automatic Control}, 
  title={An Input-Output Approach to Incremental Feedback Stability}, 
  year={2025},
  volume={},
  number={},
  pages={1-8},
  doi={10.1109/TAC.2025.3608259}}

@article{GENESIO1992,
title = {Harmonic balance methods for the analysis of chaotic dynamics in nonlinear systems},
journal = {Automatica},
volume = {28},
number = {3},
pages = {531-548},
year = {1992},
issn = {0005-1098},
doi = {https://doi.org/10.1016/0005-1098(92)90177-H},
url = {https://www.sciencedirect.com/science/article/pii/000510989290177H},
author = {R. Genesio and A. Tesi}
}

@BOOK{Curtain91,
 author = {R. F. Curtain and H. Zwart},
 title = {An introduction to infinite-dimensional linear systems theory},
 publisher = {Springer-Verlag},
 year = {1991},
}

@ARTICLE{Jonsson03,
  author={Jonsson, U.T. and Chung-Yao Kao and Megretski, A.},
  journal={IEEE Transactions on Circuits and Systems I: Fundamental Theory and Applications}, 
  title={Analysis of periodically forced uncertain feedback systems}, 
  year={2003},
  volume={50},
  number={2},
  pages={244-258},
  doi={10.1109/TCSI.2002.808218}}
